\documentclass[12pt]{article}
\usepackage{graphicx} % Required for inserting images
\usepackage{hyperref}
\usepackage{url}
\usepackage{listings,jvlisting}
\usepackage{amssymb}
\usepackage{tabularx}
\usepackage{amsmath}
\usepackage{mdframed}
\usepackage{graphicx}
\usepackage{setspace}
\usepackage{here}
\usepackage{authblk}
\usepackage{mathtools}
\usepackage{amsthm}
\usepackage{cite}
\usepackage{lscape}
\usepackage{xr}
\usepackage[authoryear]{natbib}
\usepackage{booktabs}
\usepackage{longtable}

\newtheorem{theorem}{Theorem}
\newtheorem{definition}{Definition}

\newtheorem{lemma}{Lemma}
\newtheorem{proposition}{Proposition}

\makeatletter

\newcommand*{\addFileDependency}[1]{% argument=file name and extension
\typeout{(#1)}% latexmk will find this if $recorder=0
\@addtofilelist{#1}
\IfFileExists{#1}{}{\typeout{No file #1.}}
}\makeatother

\def\bbe{{\text{\boldmath $\beta$}}}

\def\bth{{\text{\boldmath $\theta$}}}

\def\b{{\text{\boldmath $b$}}}

\def\e{{\text{\boldmath $e$}}}

\def\v{{\text{\boldmath $v$}}}

\def\x{{\text{\boldmath $x$}}}
\def\y{{\text{\boldmath $y$}}}

\def\A{{\text{\boldmath $A$}}}
\def\B{{\text{\boldmath $B$}}}

\def\D{{\text{\boldmath $D$}}}

\def\H{{\text{\boldmath $H$}}}

\def\P{{\text{\boldmath $P$}}}
\def\Q{{\text{\boldmath $Q$}}}

\def\S{{\text{\boldmath $S$}}}
\def\T{{\text{\boldmath $T$}}}
\def\U{{\text{\boldmath $U$}}}
\def\V{{\text{\boldmath $V$}}}
\def\W{{\text{\boldmath $W$}}}
\def\X{{\text{\boldmath $X$}}}
\def\Y{{\text{\boldmath $Y$}}}
\def\Z{{\text{\boldmath $Z$}}}

\begin{document}
\title{Differentially Private One-Shot Federated Inference for Linear Mixed Models via Lossless Likelihood Reconstruction}
\author[1]{Keisuke Hanada\footnote{Department of Biostatistics, Wakayama Medical University, Kimiidera, Wakayama, 641-8509, Japan \quad E-Mail: khanada@wakayama-med.ac.jp}}
\author[1]{Toshio Shimokawa}
\author[2]{Kazushi Maruo}
\affil[1]{Department of Biostatistics, Faculty of Medicine, Wakayama Medical University}
\affil[2]{Department of Biostatistics, Institute of Medicine, University of Tsukuba}
\date{}

\maketitle
\abstract{\noindent
One-shot federated learning enables multi-site inference with minimal communication. However, sharing summary statistics can still leak sensitive individual-level information when sites have only a small number of patients. In particular, shared cross-product summaries can reveal patient-level covariate patterns under discrete covariates.
Motivated by this concern, this study proposes a differentially private one-shot federated inference framework for linear mixed models with a random-intercept working covariance. The method reconstructs the pooled likelihood from site-level summary statistics and applies a Gaussian mechanism to perturb these summaries, ensuring a site-level differential privacy. Cluster-robust variance estimators are developed that are computed directly from the privatized summaries. Robust variance provides valid uncertainty quantification even under covariance mis-specification.
Under a multi-site asymptotic regime, the consistency and asymptotic normality of the proposed estimator are established and the leading-order statistical cost of privacy is characterized. 
Simulation studies show that moderate privacy noise substantially reduces reconstruction risk while maintaining competitive estimation accuracy as the number of sites increases. However, very strong privacy settings can lead to unstable standard errors when the number of sites is limited. An application using multi-site COVID-19 testing data demonstrates that meaningful privacy protection can be achieved with a modest loss of efficiency.
\par\vspace{4mm}
}

\section{Introduction}

Protecting patient-level information while enabling valid statistical inference across multiple sites is a central challenge in contemporary biomedical research.
A site with only three patients and three binary covariates can be uniquely reconstructed from its Gram matrix $X^\top X$. 
Without constraints, infinitely many matrices share the same quadratic form.
However, discrete covariate restrictions collapse the feasible set, making individual-level data identifiable up to row permutations.
This simple observation reveals a fundamental tension in one-shot federated learning (FL): summary statistics may preserve likelihood information; however, they still expose sensitive individual data when within-site sample sizes are small.
For example, consider a site with three patients and three binary covariates (Result, Sex, and Drive-through testing) from the COVID-19 testing dataset \citep{higgins2021medicaldata}.
Suppose that the site shares only the Gram matrix:
\[
X^\top X =
\begin{pmatrix}
1 & 0 & 0 \\
0 & 1 & 0 \\
0 & 0 & 0
\end{pmatrix}.
\]
From the diagonal entries, it can immediately be infered that exactly one patient tested positive, one patient was male, and no patient underwent drive-through testing. 
Under binary constraints, solving a $0$--$1$ feasibility problem \citep{fischetti2005feasibility,berthold2019ten} reveals that this Gram matrix admits a unique covariate configuration up to the permutation (see Appendix~A for the formulation and verification). 
Thus, even though individual-level records are not directly shared, sensitive attributes may become recoverable when within-site sample sizes are small.

FL has emerged as an important paradigm for analyzing multi-site data without transferring individual-level records \citep{mcmahan2017communication, kairouz2021advances}. 
In the medical domain, FL has been increasingly adopted to enable collaborative analysis across hospitals while preserving patient confidentiality \citep{rieke2020future,li2020review}. 
\cite{li2024centralized} further highlight the importance of federated approaches for clinical data analysis and the trade-offs between centralized and decentralized modeling strategies.
However, in practice, iterative communication between sites and a central server can be operationally burdensome in clinical settings. 
\cite{yan2023privacy} has extended federated learning to generalized linear mixed models using iterative communication schemes for correlated electronic health record data. 
This limitation has motivated the development of one-shot FL methods, in which each site shares summary statistics only once \citep{limpoco2025linear}. 
\cite{wu2025cola} has also developed one-shot and lossless distributed algorithms for generalized linear models, sometimes combined with cryptographic protection such as homomorphic encryption.
For linear mixed models (LMMs), distributed likelihood reconstruction enables exact recovery of the pooled likelihood from site-level sufficient statistics \citep{luo2022dlmm}.
These approaches demonstrate that summary statistics are sufficient for exact likelihood-based estimations in LMMs.
Despite their statistical efficiency, summary-based one-shot methods do not inherently guarantee privacy. 
As illustrated in the example above, when within-site sample sizes are small and covariates are discrete, Gram matrices may admit unique integer solutions that enable the reconstruction of the original covariates. 
Such scenarios are common in multi-center medical studies, in which a substantial fraction of sites may contribute to only a few patients. 
An alternative approach is to share site-specific parameter estimates and their covariance matrices, as in the aggregate-estimating equations \citep{lin2011aggregated} and Bayesian federated inference \citep{jonker2024bayesian,jonker2025bayesian}. 
Although such methods reduce reconstruction risk, their theoretical guarantees typically require large within-site sample sizes. 
However, in many medical applications, the asymptotic validity relies on the number of sites rather than on the number of patients per site.

These considerations suggest that a viable one-shot procedure should satisfy the following three properties simultaneously: (i) a likelihood-based estimation that fully exploits finite-sample information within each site, (ii) formal privacy protection against reconstruction, and (iii) valid inference under multi-site asymptotics. This is particularly important when site-level sample sizes are small, where methods relying on large within-site asymptotics may lose their validity. Likelihood-based reconstruction provides a principled way to retain finite-sample information from each site while enabling multi-site asymptotic analysis.
However, these components have largely been developed in isolation. Existing one-shot methods for linear mixed models achieve exact likelihood reconstruction but do not provide formal privacy guarantees \citep{luo2022dlmm}, whereas privacy-preserving approaches for multi-site data analysis typically rely on iterative communication, approximate inference, or large within-site sample sizes \citep{yan2023privacy,jonker2024bayesian,wu2025cola}.
To address this gap, we incorporate differential privacy, a formal framework for quantifying privacy guarantees \citep{chaudhuri2011differentially,dwork2014algorithmic}, into the distributed likelihood framework for linear mixed models. By combining exact summary-based likelihood reconstruction with privacy-preserving perturbation of site-level summaries, the proposed approach enables valid statistical inference under a multi-site regime without access to individual-level data.

The contributions of this study are threefold:
(i) An exact one-shot likelihood reconstruction for LMMs is extended to a formally privatized setting by perturbing site-level sufficient statistics via a Gaussian mechanism.
(ii) Cluster-robust variance estimators are developed that can be computed from perturbed summaries. Robust variance ensures a valid inference even under model misspecification.
(iii) Under a multi-site asymptotic regime, the consistency and asymptotic normality of the proposed estimator are established and the statistical cost of privacy is quantified.
Through simulations and an application to multi-site COVID-19 testing data, this study demonstrates that the proposed method substantially mitigates reconstruction risk while maintaining accurate inference, particularly in settings with small within-site sample sizes.

The remainder of the paper is organized as follows. Section 2 introduces the one-shot estimation for LMMs and reviews likelihood reconstruction from site-level summary statistics. Section 3 presents the proposed differentially private one-shot federated inference procedure, including the privacy mechanism and estimation method. Section 4 establishes theoretical properties of the proposed estimator under a multi-site asymptotic regime. Section 5 reports simulation studies evaluating estimation accuracy, privacy protection, and inference performance. Section 6 illustrates the proposed method through an application to multi-site clinical data. Section 7 concludes the paper. Additional technical details and supplementary results are provided in the Web appendix.

\section{Lossless one-shot estimation for LMMs}
\label{sec-method}

\subsection{LMM and individual participant data likelihood}
Let $\y$ be an observable $N$-variate vector and $\X$ be an $N \times p$ known matrix of covariates. Assume the LMM \citep{laird1982random}:
\begin{align}
\label{eq-model}
    \y &= \X \bbe + \Z \b + \varepsilon,
\end{align}
where $\bbe\in\mathbb{R}^p$ and $\Z\in\mathbb{R}^{N\times K}$ is the known site-membership matrix (that is, $Z_{ik}=1$ if observation $i$ belongs to site $k$, and $0$ otherwise). We assume a random-intercept working covariance:
\[
\b \sim N(\mathbf{0}, \tau^2 I_K),\quad
\varepsilon \sim N(\mathbf{0}, \sigma^2 I_N),\quad
\b \perp \varepsilon,
\]
where $\mathrm{Cov}(\y\mid \X)=\Sigma=\sigma^2 I_N + \tau^2 \Z\Z^\top$.
The log-likelihood function for the model parameters $(\bbe, \sigma^2, \tau^2)$ based on individual participant data (IPD) is given by
\begin{align*}
    l_{IPD}(\bbe, \sigma^2, \tau^2) &= -\frac{1}{2} \left\{ \log |\Sigma| + (\y - \X \bbe)^\top \Sigma^{-1} (\y - \X \bbe) \right\}.
\end{align*}
As $\Z$ is the site-membership matrix, the covariance matrix admits a block-diagonal representation:
\[
\Sigma=\mathrm{blockdiag}(\Sigma_1,\dots,\Sigma_K),
\quad
\Sigma_k=\sigma^2 I_{n_k}+\tau^2 \mathbf 1_{n_k}\mathbf 1_{n_k}^\top,
\]
where $n_k$ denotes the sample size of site $k$ and $\mathbf 1_{n_k}$ is an $n_k$-dimensional vector of ones.

Although LMMs allow for more general covariance structures (for example, \cite{kubokawa2021general}), a random-intercept working covariance is adopted to enable lossless likelihood reconstruction from site-level summaries.
This structure yields a closed-form block representation of $\Sigma_k$, which is essential for a one-shot federated estimation.

\subsection{Lossless reconstruction from summaries}
\label{sec-lossless-likelihood}
The likelihood-based estimation under model \eqref{eq-model} typically requires IPD $(\y, \X)$.
However, sharing IPD across sites may not be feasible because of privacy and governance reasons.
To address this limitation, likelihood reconstruction using only site-level sufficient statistics is considered, as proposed by \cite{luo2022dlmm}.
The theoretical justification for sufficiency and likelihood equivalence is provided in Section~\ref{sec-theory}.
Here, the concrete one-shot estimation procedure is described.

\paragraph{Site-level summaries}

Each site $k$ shares two matrices $(\S_k, \T_k)$ defined as follows:
\begin{align}
\label{eq-summary-statistics}
\S_k &=
\begin{pmatrix}
(\S_k)_{\y\y} & (\S_k)_{\y\X} \\
(\S_k)_{\X\y} & (\S_k)_{\X\X}
\end{pmatrix}
=
\begin{pmatrix}
\y_k^\top \y_k & \y_k^\top \X_k \\
\X_k^\top \y_k & \X_k^\top \X_k
\end{pmatrix}, \\[1ex]
\T_k &=
\begin{pmatrix}
(\T_k)_{\y\y} & (\T_k)_{\y\X} \\
(\T_k)_{\X\y} & (\T_k)_{\X\X}
\end{pmatrix}
=
\begin{pmatrix}
\y_k^\top \mathbf{1}_{n_k}\mathbf{1}_{n_k}^\top \y_k &
\y_k^\top \mathbf{1}_{n_k}\mathbf{1}_{n_k}^\top \X_k \\
\X_k^\top \mathbf{1}_{n_k}\mathbf{1}_{n_k}^\top \y_k &
\X_k^\top \mathbf{1}_{n_k}\mathbf{1}_{n_k}^\top \X_k
\end{pmatrix}.
\end{align}
Here, $n_k$ denotes the sample size of site $k$.
The matrices $\S_k$ and $\T_k$ contain all the quadratic forms required to reconstruct the likelihood under the random-intercept model and are constructed from site-level data $(\y_k, \X_k)$.

\paragraph{One-shot likelihood reconstruction}

Using only $\{\S_k, \T_k\}_{k=1}^K$, the log likelihood function for $(\bbe, \sigma^2, \tau^2)$ can be expressed as follows:
\begin{align}
\label{eq-oneshot-ml}
l_{ML}(\bbe, \sigma^2, \tau^2; \{\S_k, \T_k\}_{k=1}^K)
&= - \frac{1}{2} \sum_{k=1}^K
\Bigg[
\log (\sigma^2)^{n_k-1}(\sigma^2+n_k \tau^2)
\nonumber\\
&\quad
+ \frac{1}{\sigma^2}
\begin{pmatrix}
1 \\
-\bbe
\end{pmatrix}^\top
\left\{
\S_k
-
\frac{\tau^2}{\sigma^2 + n_k \tau^2}
\T_k
\right\}
\begin{pmatrix}
1 \\
-\bbe
\end{pmatrix}
\Bigg].
\end{align}
Importantly, this expression is algebraically identical to the IPD-based likelihood obtained from $(\y, \X)$, such that no information loss occurs under the exact summaries.
The restricted maximum likelihood (REML) function can likewise be written as follows:
\begin{align}
\label{eq-oneshot-reml}
l_{REML}(\sigma^2, \tau^2; \{\S_k, \T_k \}_{k=1}^K)
&=
l_{ML}(\hat{\bbe}(\sigma^2,\tau^2), \sigma^2, \tau^2; \{\S_k, \T_k \}_{k=1}^K)
\nonumber\\
&\quad
-
\frac{1}{2} \log
\left|
\frac{1}{\sigma^2}
\sum_{k=1}^K
\left\{
(\S_k)_{\X \X}
-
\frac{\tau^2}{\sigma^2 + n_k \tau^2}
(\T_k)_{\X \X}
\right\}
\right|,
\end{align}
where $\hat{\bbe}(\sigma^2,\tau^2)$ denotes the maximizer of the likelihood with respect to $\bbe$ for fixed $(\sigma^2,\tau^2)$, which can be expressed in closed form using the summary statistics.
This establishes a one-shot federated estimator that requires only a single communication of the site-level summaries.
Although the REML function can also be reconstructed from site-level summaries, the additional determinant term involves aggregated cross-site information, which becomes numerically unstable under Gaussian perturbation.
In particular, determinant amplification can lead to excessive noise when DP is imposed.
Therefore, subsequent developments focus on maximum likelihood (ML)-based inference, which allows for a more stable privacy-preserving extension.

\subsection{Proposed cluster-robust inference from summaries}
\label{sec-robust-variance}

Under the working random-intercept covariance, \cite{luo2022dlmm} showed that the pooled likelihood can be reconstructed exactly from site-level quadratic summaries, as described in Section \ref{sec-lossless-likelihood}.
Although point estimation follows directly from this reconstruction, model mis-specifications in the mean or covariance structure may occur.
To ensure a valid inference under potential misspecification, the cluster-robust variance estimator $\hat{\V}_{\mathrm{CR0}}$ \citep{white1980heteroskedasticity, liang1986longitudinal} is adopted, given by:
\begin{align*}
    \hat{\V}_{\mathrm{CR0}}
    &= (\X^\top \Sigma^{-1} \X)^{-1} \left( \sum_{k=1}^K \X_k^\top \Sigma_k^{-1} \hat{\e}_k \hat{\e}_k^\top \Sigma_k^{-1} \X_k \right) (\X^\top \Sigma^{-1} \X)^{-1},
\end{align*}
where $\hat{\e}_k=\y_k - \X_k \hat{\bbe}$ denotes the residual vector of site $k$.

The key question is whether $\hat{\V}_{\mathrm{CR0}}$ can be computed using site-level summaries only.
To address this question, the following result is established.
\begin{theorem}[Likelihood reconstruction from site-level summaries]
\label{th-consistent-non-dp}
Under the LMM \eqref{eq-model}, the collection of quadratic summaries $\{\S_k, \T_k\}_{k=1}^K$ permits the exact reconstruction of the pooled ML and REML functions.
In particular, both $l_{ML}(\bbe, \sigma^2, \tau^2)$ and $l_{REML}(\bbe, \sigma^2, \tau^2)$ can be expressed in terms of $\{\S_k, \T_k\}_{k=1}^K$.
Furthermore, the cluster-robust variance estimator $\hat{\V}_{\mathrm{CR0}}$ admits an equivalent representation based solely on $\{\S_k, \T_k\}_{k=1}^K$, as follows:
\begin{align*}
    \hat{\V}_{\mathrm{CR0}} &= \left(\sum_{k=1}^K \W_k\right)^{-1} \left( \sum_{k=1}^K \P_k \right) \left(\sum_{k=1}^K \W_k\right)^{-1},
\end{align*}
where the components $\W_k$, $\Q_k$, and $\P_k$ are defined as follows:
\begin{align*}
    \W_k &= \X_k^\top \Sigma_k^{-1} \X_k = \frac{1}{\sigma^2} (\S_k)_{\X \X} - \frac{\tau^2}{\sigma^2(\sigma^2+n_k \tau^2)} (\T_k)_{\X \X}, \\
    \Q_k &= \X_k^\top \Sigma_k^{-1} \y_k = \frac{1}{\sigma^2} (\S_k)_{\X \y} - \frac{\tau^2}{\sigma^2(\sigma^2+n_k \tau^2)} (\T_k)_{\X \y}, \\
    \P_k &= (\Q_k - \W_k \hat{\bbe}) (\Q_k - \W_k \hat{\bbe})^\top.
\end{align*}
\end{theorem}

Theorem \ref{th-consistent-non-dp} establishes a finite-sample likelihood equivalence, and the likelihood function and cluster-robust variance computed from individual-level data can be losslessly reconstructed from the site-level summaries $\{\S_k, \T_k\}$.
The consistency and asymptotic normality of the ML and REML estimators still require the usual regularity conditions; see \citet{jiang1996reml}.

Although the CR0 variance estimator provides a natural sandwich-based correction under model misspecification, it is well known to exhibit downward bias when the number of clusters is small. 
Therefore, scalar small-sample adjustments can be applied by multiplying $\hat{\V}_{\mathrm{CR0}}$ with a constant factor. 
In particular, the CR1, CR1p, and CR1S corrections correspond to multiplication by $K/(K-1)$, $K/(K-p)$, and $[K(N-1)]/[(K-1)(N-p)]$, respectively, without requiring IPD.
By contrast, CR2 and CR3 require cluster-specific leverage adjustments based on individual-level design and residual information, and thus cannot be implemented from quadratic summaries alone \citep{cameron2015practitioner, pustejovsky2018small}.

\section{Differentially private one-shot estimation}

\subsection{DP mechanism}
\label{sec-define-dp}

Sharing the site-level summaries $\{\S_k, \T_k\}$ enables ML estimation without access to IPD.
However, releasing these summaries without modification may pose a reconstruction risk because the original covariates can potentially be recovered through matrix factorization (see Appendix A). 
To mitigate the risk of individual data reconstruction from summaries, a DP framework is adopted.
As the primary concern is the reconstruction of matrix-valued quantities, such as $\X_k^\top \X_k$, proximity is measured using the Frobenius norm. 
For an $n \times m$ matrix $\A$, the Frobenius norm is defined as follows:
\[
\|\A\|_F = \sqrt{\sum_{i=1}^{n}\sum_{j=1}^{m}A_{ij}^2},
\]
where $A_{ij}$ denotes the $(i,j)$-th entry of $\A$.

Define two datasets $\D$ and $\D'$ to be adjacent if they differ by exactly one individual record, i.e., there exists an index $i$ such that $(X_{i}, y_{i}) \ne (X'_{i}, y'_{i})$, while all other observations are identical.
Under this adjacency concept, DP is defined as follows.

\begin{definition}[Differential privacy; Definition 2.4 in \cite{dwork2014algorithmic}]
\label{def-dp}
Let $\mathcal{X}$ denote the space of datasets consisting of observations $(X,y)$, and $\mathcal{Y}$ be the output space of the randomized mechanism.
A randomized algorithm $\mathcal{M} : \mathcal{X} \to \mathcal{Y}$ is said to be $(\varepsilon,\delta)$-differentially private if for all adjacent $\D,\D' \in \mathcal{X}$ and for all measurable sets $\mathcal{S} \subseteq \mathcal{Y}$,
\[
P\big(\mathcal{M}(\D) \in \mathcal{S}\big)
\le
e^{\varepsilon}
P\big(\mathcal{M}(\D') \in \mathcal{S}\big)
+ \delta,
\]
where the probability is taken over the internal randomness of $\mathcal{M}$.
\end{definition}

A Gaussian perturbation mechanism is employed to construct a differentially private mechanism.
Let $Q : \mathcal{X} \to \mathbb{R}^{m\times n}$ be a matrix-valued query. In our setting, the query $Q$ corresponds to the site-level summaries $(S_k, T_k)$.
To ensure finite global sensitivity, we assume that each individual record is bounded, i.e., the covariates and outcome are restricted to prespecified finite ranges.
The global sensitivity of $Q$ under the Frobenius norm is defined as follows:
\[
\Delta_F
= \sup_{\text{adjacent }\D,\D'} \|Q(\D) - Q(\D')\|_F.
\]
The Gaussian mechanism is defined as $\mathcal{M}_g(\D) = Q(\D) + \U$, where $\U$ is a random matrix with entries that are i.i.d. $N(0,\sigma_{DP}^2)$ with $\sigma_{DP} = \Delta_F \sqrt{2 \log (1.25/\delta)} / \varepsilon$.

\begin{proposition}[Gaussian mechanism]
\label{th-gaussian-mechanism-dp}
Let $Q$ be a matrix-valued query with global sensitivity $\Delta_F$.
Then the mechanism
\[
\mathcal{M}_g(\D)=Q(\D)+\U,
\quad
U_{ij}\sim N(0,\sigma_{DP}^2),
\]
is $(\varepsilon,\delta)$-differentially private.
\end{proposition}

In our setting, the query involves cross-product summaries including $\X^\top \X$, $\X^\top \y$, and $\y^\top \y$.
For example, consider the query $Q(\D)=\X^\top \X$.
Suppose that each covariate vector $\x_i \in \{0,1\}^p$ and that two design matrices are adjacent if they differ in exactly one row.
Let $\X$ and $\X'$ differ only in the $i$th row and be denoted by $\x$ and $\x'$, respectively.
Then $Q(\X)-Q(\X') = \x \x^\top - \x' \x'^\top$.
Using the triangle inequality and the identity $\|\x \x^\top\|_F=\|\x\|_2^2$ yields $\|Q(\X)-Q(\X')\|_F \le \|\x\|_2^2 + \|\x'\|_2^2$.
As $\|\x\|_2^2 \le p$ for binary covariates, it follows that $\Delta_F \le 2p$.

\subsection{Proposed differentially private one-shot estimation}

To mitigate the reconstruction risks, the Gaussian perturbation mechanism introduced in Section \ref{sec-define-dp} is applied to the shared summaries. Specifically, each site releases $\tilde{\S}_k = \S_k + \U_{k1}$ and $\tilde{\T}_k = \T_k + \U_{k2}$, where the entries of $\U_{k1}$ and $\U_{k2}$ are independent zero-mean Gaussian random variables with variance calibrated to achieve $(\varepsilon, \delta)$-DP. 

Using noisy summaries $\{ \tilde{\S}_k, \tilde{\T}_k \}_{k=1}^K$, the perturbed likelihood function is defined as follows:
\begin{align}
\label{eq-ll-dp}
    l_{DP}(\bbe, \sigma^2, \tau^2; \{\tilde{\S}_k, \tilde{\T}_k\}_{k=1}^K )
    =
    l_{ML}(\bbe, \sigma^2, \tau^2; \{\tilde{\S}_k, \tilde{\T}_k\}_{k=1}^K).
\end{align}
Note that $l_{DP}$ is a perturbed objective function due to the injected noise and is not the true likelihood. The DP-perturbed ML estimator is defined as follows:
\[
(\hat{\bbe}_{DP}, \hat{\sigma}_{DP}^2, \hat{\tau}_{DP}^2)
=
\arg\max_{\bbe, \sigma^2, \tau^2}
l_{DP}(\bbe, \sigma^2, \tau^2; \{\tilde{\S}_k, \tilde{\T}_k\}_{k=1}^K).
\]
The DP-based ML estimator and its robust variance estimator are given by:
\begin{align}
\label{eq-estimator-dp}
    \hat{\bbe}_{DP}
    &= \left( \sum_{k=1}^K \tilde{\W}_k \right)^{-1}
       \left( \sum_{k=1}^K \tilde{\Q}_k \right), \\
\label{eq-variance-dp}
    \hat{\V}_{DP}
    &= \left(\sum_{k=1}^K \tilde{\W}_k\right)^{-1}
       \left( \sum_{k=1}^K \tilde{P}_k \right)
       \left(\sum_{k=1}^K \tilde{\W}_k\right)^{-1},
\end{align}
where
\begin{align*}
    \tilde{\W}_k
    &= \frac{1}{\hat{\sigma}_{DP}^2} (\tilde{\S}_k)_{\X \X}
       - \frac{\hat{\tau}_{DP}^2}
              {\hat{\sigma}_{DP}^2
               (\hat{\sigma}_{DP}^2 + n_k \hat{\tau}_{DP}^2)}
         (\tilde{\T}_k)_{\X \X}, \\
    \tilde{\Q}_k
    &= \frac{1}{\hat{\sigma}_{DP}^2} (\tilde{\S}_k)_{\X \y}
       - \frac{\hat{\tau}_{DP}^2}
              {\hat{\sigma}_{DP}^2
               (\hat{\sigma}_{DP}^2 + n_k \hat{\tau}_{DP}^2)}
         (\tilde{\T}_k)_{\X \y}, \\
    \tilde{\P}_k
    &= (\tilde{\Q}_k - \tilde{\W}_k \hat{\bbe}_{DP})
       (\tilde{\Q}_k - \tilde{\W}_k \hat{\bbe}_{DP})^\top.
\end{align*}
The DP-perturbed ML estimator $(\hat{\bbe}_{DP}, \hat{\sigma}_{DP}^2, \hat{\tau}_{DP}^2)$ is defined as the joint maximizer of $l_{DP}$.
At the maximizer, $\hat{\bbe}_{DP}$ admits a closed-form representation of \eqref{eq-estimator-dp}.
The variance estimator $\hat{\V}_{DP}$ corresponds to the cluster-robust (CR0) sandwich variance.

To preserve the symmetry of $\tilde{\S}_k$ and $\tilde{\T}_k$, post-processing symmetrization is applied, where $\U_{k1} \leftarrow (\U_{k1} + \U_{k1}^\top)/2$ and $\U_{k2} \leftarrow (\U_{k2} + \U_{k2}^\top)/2$.
The post-processing step does not affect the DP guarantee \citep{dwork2014algorithmic}.

\section{Asymptotic theory}
\label{sec-theory}

This section establishes the asymptotic properties of the proposed estimator under DP.
All proofs are provided in Appendix \ref{app-proofs}.

A multi-site asymptotic regime is considered in which the number of sites $K \to \infty$ and within-site sample sizes $n_k$ are allowed to remain bounded.
Under this framework, the DP-perturbed objective function remains a sum of independent site-level contributions and therefore fits into the general theory of M-estimation.

First, the notation is introduced.
Let $\mathcal{D}_k = \{n_k, \tilde{\S}_k, \tilde{\T}_k\}$ and the site-wise contribution is defined as follows:
\begin{align*}
    m_{\bth}(\mathcal{D}_k) &= - \frac{1}{2} \left[ \log (\sigma^2)^{n_k-1}(\sigma^2 + n_k \tau^2) + \frac{1}{\sigma^2} 
    \begin{pmatrix}
        1 \\
        -\bbe
    \end{pmatrix}^\top
    \left\{ \tilde{\S}_k - \frac{\tau^2}{\sigma^2 + n_k \tau^2} \tilde{\T}_k \right\}
    \begin{pmatrix}
        1 \\
        -\bbe
    \end{pmatrix}
    \right].
\end{align*}
The full DP objective is expressed as $l_{DP}(\bth) = \sum_{k=1}^K m_{\bth}(\mathcal{D}_k)$, where $\bth=(\bbe, \sigma^2, \tau^2)$.
Hence, this estimator is an M-estimator with independent site-level contributions.
The gradient and Hessian of $m_{\bth}(\mathcal{D}_k)$ are defined as follows:
\begin{align*}
    \dot{m}_{\bth}(\mathcal{D}_k) = \frac{\partial}{\partial \bth} m_{\bth}(\mathcal{D}_k), \quad
    \ddot{m}_{\bth}(\mathcal{D}_k) = \frac{\partial^2}{\partial\bth \partial\bth^\top} m_{\bth}(\mathcal{D}_k).
\end{align*}
Let $P$ denote the probability measure under model \eqref{eq-model}.
The following assumptions are imposed:

\begin{enumerate}
    \item[(A1)] Parameter space compactness:
    The parameter space $\Theta \subset \mathbb{R}^{p+2}$ is compact and of the form $\Theta = \mathcal{B}_\bbe \times [\underline\sigma^2, \overline\sigma^2] \times [0, \overline\tau^2]$, where $\mathcal{B}_\bbe \subset \mathbb{R}^{p}$ is compact and $0 < \underline\sigma^2 < \overline\sigma^2 < \infty$, $0 < \overline\tau^2 < \infty$.
    \item[(A2)] Identifiability:
    The expected log-likelihood $M(\bth)=E[m_{\bth}(\mathcal{D}_k)]$ is uniquely maximized at the true (or pseudo-true) parameter $\bth_0 = (\bbe_0, \sigma_0^2, \tau_0^2)$, and $\bth_0$ lies in the interior of $\Theta$.
    \item[(A3)] Site-level independence and uniform law of large numbers: Under the working model, the site-level contributions $\{m_{\bth}(\mathcal{D}_k)\}_{k=1}^K$ are independent and identically distributed, and the non-DP log-likelihood satisfies $M_K^{ML}(\bth) = K^{-1} l_{ML}(\bth)$ and
    $\sup_{\bth \in \Theta} |M_K^{ML}(\bth) - M(\bth)| \xrightarrow{P} 0$.
    \item[(A4)] Independence and moments of DP noise:
    Let $\tilde{\S}_k=\S_k + \U_{k1}$ and $\tilde{\T}_k = \T_k + \U_{k2}$, where
    \begin{enumerate}
        \item[(i)] $\{\U_{k1},\U_{k2}\}_{k=1}^K$ are independent across $k$ and independent of $\{\S_k,\T_k\}$,
        \item[(ii)] each entry has zero mean and finite variance (e.g.\ $(\U_{k\ell})_{ij} \sim N(0,\sigma_{DP}^2)$ i.i.d.).
    \end{enumerate}
    \item[(A5)] Moment and non-singularity conditions:
    \begin{enumerate}
        \item[(i)] $P\|\dot{m}_{\bth_0}(\mathcal{D}_k)\|^2 < \infty$,
        \item[(ii)] $P[\dot{m}_{\bth_0}(\mathcal{D}_k) \dot{m}_{\bth_0}(\mathcal{D}_k)^\top]$ is nonsingular.
    \end{enumerate}
\end{enumerate}

Assumptions (A1)--(A3) and (A5) are standard regularity conditions that ensure the consistency and asymptotic normality of the non-private ML estimator.
Assumption (A4) characterizes the DP perturbation. It requires only independence and finite second moments; therefore, it allows noise distributions beyond the Gaussian mechanism.

Under these conditions, the following consistency and asymptotic normality results hold.

\begin{theorem}
\label{th-consistent-dp}
Under the assumption (A1)--(A4), the estimator $\hat{\bth}_{DP}$, defined as follows:
\begin{align*}
    \hat{\bth}_{DP} &= \arg \max_{\bth \in \Theta} l_{DP}(\bbe, \sigma^2, \tau^2; \{\tilde{\S}_k, \tilde{\T}_k \}_{k=1}^K),
\end{align*}
converges in probability to $\bth_0$ as $K \to \infty$.
\end{theorem}

The consistency follows from the uniform convergence of the DP-perturbed objective to its population counterpart.
As the perturbation has zero mean and finite variance, its average contribution vanishes as $K \to \infty$ and therefore does not affect the probability limit.

\begin{theorem}
\label{th-asymp-dp}
Under the assumption (A1)--(A5), $\sqrt{K} (\hat{\bth}_{DP} - \bth_0) \xrightarrow{d} N(\mathbf{0}, \V_{DP})$ as $K \to \infty$, where
\begin{align*}
    \V_{DP} &= \left( P[\ddot{m}_{\bth_0}(\mathcal{D}_k)] \right)^{-1}
    \left( P[\dot{m}_{\bth_0}(\mathcal{D}_k) \dot{m}_{\bth_0}(\mathcal{D}_k)^\top] \right)
    \left( P[\ddot{m}_{\bth_0}(\mathcal{D}_k)] \right)^{-1}.
\end{align*}
Furthermore, the sandwich estimator $\hat{\V}_{DP}$, defined in \eqref{eq-variance-dp}, is a consistent estimator of the $(\bbe,\bbe)$-block of $\V_{DP}$ as follows:
\begin{align*}
    \hat{\V}_{DP} \xrightarrow{P} (\V_{DP})_{\bbe\bbe}, \quad K \to \infty.
\end{align*}
\end{theorem}

Asymptotic normality is a direct consequence of the standard M-estimation theory.
Compared with the non-private estimator, the matrix:
\[
P[\dot{m}_{\bth_0}(\mathcal{D}_k) \dot{m}_{\bth_0}(\mathcal{D}_k)^\top]
\]
contains an additional variance component induced by the DP perturbation, whereas the leading-order curvature term $P[\ddot{m}_{\bth_0}(\mathcal{D}_k)]$ remains unchanged.
Consequently, the DP estimator achieves a valid inference under potential working covariance mis-specifications, while incurring additional variability owing to privacy noise.

Next, the additional estimation error induced by DP perturbation is quantified.

\begin{proposition}
\label{th-dp-loss}
Under the assumptions (A1)--(A4), the expectation of $L_2$ deviation is as follows:
\begin{align*}
    E\| \hat{\bbe}_{DP} - \hat{\bbe} \|^2
    &= \operatorname{tr}(\Omega) \frac{\sigma_{DP}^2}{K} + o(K^{-1}),
\end{align*}
where $\Omega$ is a positive definite matrix that depends on second-order
structure of the model and design matrices.
\end{proposition}

In particular, under a Gaussian mechanism:
\begin{align*}
    E\| \hat{\bbe}_{DP} - \hat{\bbe} \|^2
    &= \operatorname{tr}(\Omega)
    \frac{\Delta_F^2}{K}
    \frac{2 \log(1.25 / \delta)}{\varepsilon^2}
    + o(K^{-1}).
\end{align*}
Therefore, the privacy cost is scaled to $O(\sigma_{DP}^2/K)$.
Stronger privacy protection (smaller $\varepsilon$ or $\delta$) increases the estimation error, whereas increasing the number of sites $K$ mitigates the effect of the perturbation.
When $K$ is moderately large, the additional error becomes negligible, even under relatively stringent privacy budgets.
In practical applications, the privacy parameters $(\varepsilon, \delta)$ should be selected by balancing formal privacy guarantees with the resulting estimation accuracy.

\section{Simulation studies}
This section describes two simulation studies conducted to evaluate the reconstruction risk and estimation performance of the proposed method.
The first simulation examined the reconstruction risk from shared summaries, whereas the second investigated the estimation accuracy under DP.

\subsection{Reconstruction of IPD from the Gram matrix}
\label{sec-sim-dp-check}

\paragraph{Simulation configuration}
The reconstructability of individual-level data was assessed in various privacy settings.
Specifically, a single site ($K=1$) was considered and a covariate matrix $\X_k \in \{0,1\}^{n \times p}$ was generated.
The Gram matrix $(\S_k)_{\X\X} = \X_k^\top \X_k$ was computed and perturbed using a Gaussian mechanism.
Each entry of $\X$ was generated independently from $\mathrm{Bernoulli}(0.5)$.
The sample size $n$ varied in the range $2--20$ in increments of $1$, and the dimensions $p$ were set to $3$, $5$, and $10$.
To reconstruct the IPD, the perturbed Gram matrix was rounded to the nearest integer and a 0--1 mixed integer programming problem was solved (see Appendix A), yielding the reconstructed data $\hat{\X}_k$.

\paragraph{Privacy calibration}
Under the binary design $\X_k \in \{0,1\}^{n\times p}$ and the adjacency notion that two datasets differ by one individual, the global sensitivity is $\Delta_F=2p$, as discussed in Section \ref{sec-define-dp}. 
As sensitivity increases with the dimension $p$, fixing the privacy parameter $\varepsilon$ across dimensions would result in a noise variance that grows linearly with $p$. 
However, reconstruction from the Gram matrix was formulated as a 0--1 integer programming problem with combinatorial complexity increasing rapidly with $p$. 
To reflect this structural increase in reconstruction difficulty, a dimension-adjusted privacy-allocation strategy was adopted and $\varepsilon(p) = 2p \, \varepsilon_0$ was set, where $\varepsilon_0 \in \{1,2,4,6,8,10,12,16,20\}$ was the baseline privacy level.
Parameter $\delta$ was fixed at $\delta=0.01$ for all the scenarios.
Under this specification, the Gaussian noise variance became $\sigma_{DP}^2 = 2 \log(1.25/\delta) /\varepsilon_0^2$.
Therefore, the effective perturbation was controlled solely by $\varepsilon_0$ with a fixed $\delta$.
This setup was used to examine the reconstruction risk across a wide range of privacy levels.

\paragraph{Evaluation metrics}
For each configuration, 10,000 simulation replicates were performed.
The reconstruction accuracy was evaluated using the Hamming distance between $\hat{\X}_k$ and the true matrix $\X_k$.
For the matrices $A,B \in \mathbb{R}^{m\times n}$, the Hamming distance is defined as follows:
\[
d_H(A,B) = \sum_{i=1}^{m} \sum_{j=1}^{n} \mathbb{I}(A_{ij} \neq B_{ij}),
\]
where $\mathbb{I}(\cdot)$ is the indicator function.
As the Gram matrix does not identify the row order, both $\hat{\X}_k$ and $\X_k$ were lexicographically sorted before comparison.
The reconstruction rates are defined as follows:
\begin{align*}
\text{Matrix-level reconstruction rate}
&=
\mathbb{I}\big(d_H(\hat{\X}_k,\X_k)=0\big), \\
\text{Element-level reconstruction rate}
&=
1 - \frac{d_H(\hat{\X}_k,\X_k)}{np}.
\end{align*}

\paragraph{Simulation results}
Figure \ref{fig-dp-check} presents the reconstruction rates.
When no noise was added (Ref), the reconstruction rate was high, particularly for a small $n$.
By contrast, introducing Gaussian perturbation substantially reduced the reconstruction rate, particularly for smaller values of $\varepsilon_0$ (e.g., $\varepsilon_0 \le 8$).
At the matrix level, even moderate privacy settings led to a near-zero reconstruction probability.
At the element level, the reconstruction accuracy decreased monotonically as $\varepsilon_0$ decreased.
These results indicate that the incorporation of DP effectively mitigates the risk of IPD reconstruction from the Gram matrix.

\begin{figure}[H]
    \centering
    \includegraphics[width=0.9\textwidth]{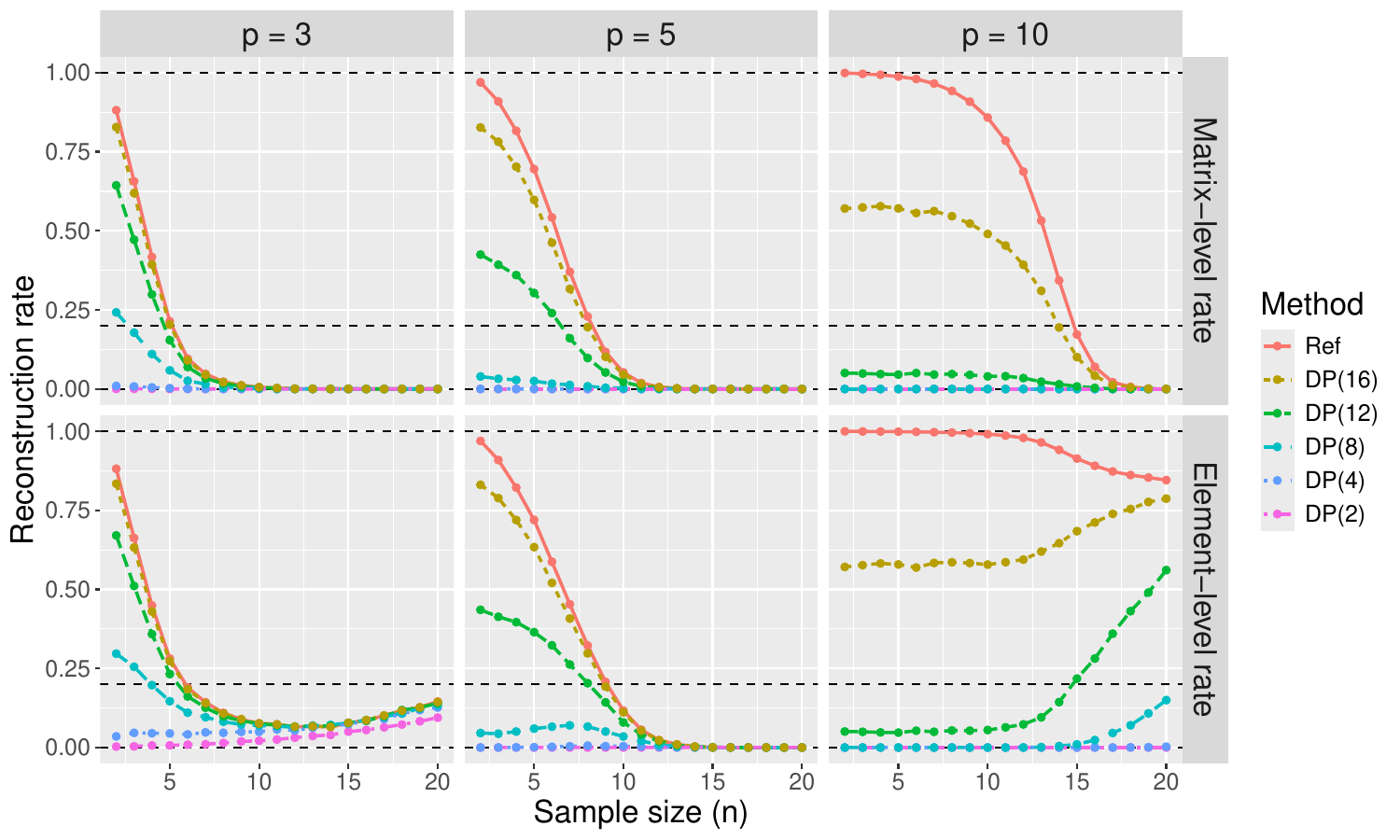}
    \caption{Matrix- and element-level reconstruction rates under various privacy settings. Ref corresponds to the case without Gaussian perturbation and DP(x) corresponds to the case of Gaussian perturbation with $\varepsilon_0=x$.}
    \label{fig-dp-check}
\end{figure}

\subsection{Comparison of estimation performance}
\label{sec-sim-estimation}

\paragraph{Simulation configuration.}
The extent to which the proposed method and existing approaches suffered from performance degradation were compared with analyses based on IPD.
Two data-generation mechanisms were considered:
\begin{enumerate}
    \item[(1)] Random-intercept model
    \begin{align*}
        y_{ki} &= \beta_0 + \sum_{m=1}^6 \beta_m x_{mki} + b_k + \epsilon_{ki}, \\
        b_k &\sim N(0, \tau^2), \quad \epsilon_{ki} \sim N(0, \sigma^2).
    \end{align*}

    \item[(2)] Random-intercept and slope model
    \begin{align*}
        y_{ki} &= \beta_0 + \sum_{m=1}^6 \beta_m x_{mki} + b_{0k} + b_{1k} x_{1ki} + \epsilon_{ki}, \\
        \begin{pmatrix}
            b_{0k} \\
            b_{1k}
        \end{pmatrix}
        &\sim
        N\!\left(
            \begin{pmatrix}
                0 \\
                0
            \end{pmatrix},
            \begin{pmatrix}
                \tau^2 & 0 \\
                0 & \tau^2
            \end{pmatrix}
        \right), \quad
        \epsilon_{ki} \sim N(0, \sigma^2).
    \end{align*}
\end{enumerate}

Among the six covariates, four were binary and two were continuous.
Specifically, $x_{mki} \sim \text{Bernoulli}(p_m)$ for $m=1,3,4,5$ with $(p_1,p_3,p_4,p_5)=(0.5,0.3,0.7,0.5)$, and $x_{mki} \sim N(0,v_m^2)$ for $m=2,6$ with $(v_2,v_6)=(1,0.5)$.
All covariates were generated independently.
The true regression coefficients were
$
(\beta_0, \beta_1, \beta_2, \beta_3, \beta_4, \beta_5, \beta_6)=(1, 0.5, 0.5, -1, -0.5, 1, -1)$,
with $\sigma^2 = 1$ and $\tau^2 = 1$.
The number of sites was set as $K=20, 50, 100, 200$.
To reflect realistic heterogeneity in multicenter studies, within-site sample sizes were generated from a mixture distribution. 
With a probability of 0.8, a within-site sample size was sampled uniformly from integers 2--10; otherwise, it was sampled uniformly from 50--100. 
Thus, 80\% of sites were small (sizes 2--10) and 20\% were large (sizes 50--100).
This design induced substantial imbalance across sites and allowed for the evaluation of robustness under heterogeneous cluster sizes.

For each data-generation mechanism, both correctly specified and mis-specified models were considered.
In the mis-specified setting, only two covariates $(x_{1ki}, x_{2ki})$ were included in the analysis model.
This setting mimiked underfitting owing to the omission of relevant covariates and allowed for the assessment of the robustness to model misspecification.
The proposed method adopts a simple random-intercept model as the working covariance structure. Therefore, when the data-generation mechanism followed a random-intercept and slope model, the working covariance was mis-specified.
Combining the two data-generation processes and the two analysis models yielded four scenarios in total.
To ensure numerical stability, $(\X,\Y)$ were standardized prior to the estimation and subsequently transformed back to the original scale.
The proposed DP-based estimation was applied to all the covariates (DP) and three privacy-sensitive covariates $x_4, x_5, x_6$ (DP2).
The dimension-adjusted privacy calibration described in Section~\ref{sec-sim-dp-check} was used throughout, with $\delta=1/N$ and $\varepsilon_0 \in \{1,2,4,6,8,10,12,16,20\}$.

\paragraph{Evaluation metrics}

For each scenario, 10,000 simulation replicates were conducted.
The following performance measures were computed:
\begin{align*}
    L_2 \text{ error from } \bbe_0 &= \| \hat{\bbe} - \bbe_0 \|, \\
    L_2 \text{ privacy cost} &= \| \hat{\bbe}_{DP} - \hat{\bbe}_{IPD} \|, \\
    \text{SE inflation factor} &= 
    \frac{\| \text{SE}(\hat{\bbe}_{DP})\|}{\|\text{SE}(\hat{\bbe}_{IPD})\|}, \\
    \text{SE calibration ratio}
    &=
    \frac{E\!\left[\widehat{\text{SE}}(\hat{\bbe})\right]}{\text{SD}(\hat{\bbe})}.
\end{align*}
The first measure evaluated the overall estimation accuracy relative to the true parameter.
The second measure isolated the additional estimation loss induced by DP.
The standard error (SE) inflation factor quantified the increase in uncertainty under privacy protection.
The SE calibration ratio indicated whether the estimated SEs were properly calibrated.
For both the SE inflation factor and SE calibration ratio, values greater than one indicated overestimation, whereas values less than one indicated underestimation.

\paragraph{Simulation results}

Figure~\ref{fig-sim-estimation} and Table~\ref{tab-sim-estimation} shows the results of the correctly specified random-intercept model. Results for the remaining scenarios are provided in Web Appendix C. Consistent with the theoretical results, the estimation performance improved as $\varepsilon_0$ and the number of sites $K$ increased. Both the $L_2$ error from $\bbe_0$ and the $L_2$ privacy cost decreased with a larger $\varepsilon_0$, indicating convergence toward the IPD-based estimator. Similarly, the privacy-induced estimation loss decreased as $K$ increased, which is consistent with the theoretical $O(K^{-1})$ decay in the privacy cost.

The SE inflation factor increased as $\varepsilon_0$ decreased. In particular, when the number of sites $K$ is small, strong privacy protection can inflate the SEs substantially. This effect was also reflected in the SE calibration ratio, indicating that reliable inference becomes difficult under stringent privacy settings.
This behavior can be explained by the fact that the DP mechanism adds noise directly to the summary statistics, with magnitude determined by the global sensitivity and the privacy parameter $\varepsilon_0$, as characterized in the Gaussian differential privacy framework \citep{dong2022gaussian}. When $K$ was small, this additional noise was not sufficiently averaged out across sites and therefore dominated the variability of the estimator. In such situations, a weaker privacy level (i.e., larger $\varepsilon_0$) may be required for maintaining reasonable estimation accuracy.

Under model mis-specification, the SE inflation factor was slightly smaller because the baseline SEs were already inflated owing to omitted covariates. Similar patterns were observed for the random-intercept and slope model. Robustness across random-effects structures can be attributed to the use of cluster-robust variance estimation, which preserves consistency when the number of sites is sufficiently large. Overall, the degradation in the estimation performance due to DP depends primarily on the privacy level $\varepsilon_0$ and number of sites $K$, while the impact of random-effects structure and model mis-specification is comparatively minor. These findings suggest that the proposed approach is most effective in settings with a sufficiently large number of sites, where the impact of DP noise can be mitigated through aggregation.

\begin{figure}[H]
    \centering
    \includegraphics[width=1.0\textwidth, page=1]{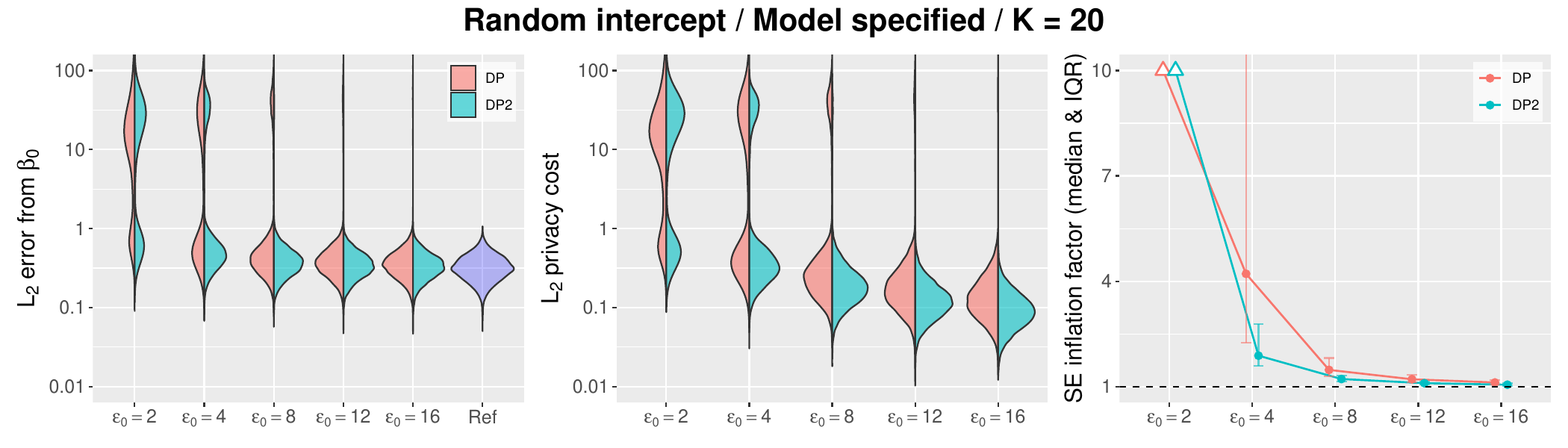}
    \includegraphics[width=1.0\textwidth, page=2]{main-figure/result-plot-all.pdf}
    \includegraphics[width=1.0\textwidth, page=3]{main-figure/result-plot-all.pdf}
    \includegraphics[width=1.0\textwidth, page=4]{main-figure/result-plot-all.pdf}
    \caption{Estimation performance results. Triangles indicate median SE inflation factors exceeding the plotting range; inter-quantile range (IQRs) are omitted in those cases.}
    \label{fig-sim-estimation}
\end{figure}

\begin{table}[H]
\centering
\caption{SE calibration ratio for $x_1$. The scenario is the correctly specified random-intercept model. CR0 denotes the CR0 cluster-robust variance estimator, and CR1p denotes the CR1p variance estimator with a small-sample correction, which are defined in Section \ref{sec-robust-variance}.}
\label{tab-sim-estimation}
\centering
\fontsize{9}{11}\selectfont
\begin{tabular}[t]{cc|ccc|ccc}
\toprule
& & \multicolumn{3}{c}{CR0} & \multicolumn{3}{c}{CR1p} \\
\cmidrule(lr){3-5} \cmidrule(lr){6-8}
$K$ & $\varepsilon_0$ & IPD & DP & DP2 & IPD & DP & DP2 \\
\midrule
20 & 2 & 0.86 & 1745.41 & 4444.54 & 1.23 & 2493.45 & 6349.34\\
20 & 4 & 0.86 & 1348.82 & 1463.55 & 1.23 & 1926.88 & 2090.79\\
20 & 8 & 0.86 & 1205.85 & 34.39 & 1.23 & 1722.64 & 49.13\\
20 & 12 & 0.86 & 308.16 & 5.82 & 1.23 & 440.23 & 8.32\\
20 & 16 & 0.86 & 49.15 & 0.86 & 1.23 & 70.21 & 1.23\\
\hline
50 & 2 & 0.95 & 3217.91 & 1724.02 & 1.08 & 3656.72 & 1959.11\\
50 & 4 & 0.95 & 1556.33 & 111.96 & 1.08 & 1768.56 & 127.23\\
50 & 8 & 0.95 & 297.79 & 0.98 & 1.08 & 338.39 & 1.11\\
50 & 12 & 0.95 & 19.61 & 0.95 & 1.08 & 22.29 & 1.08\\
50 & 16 & 0.95 & 1.01 & 0.95 & 1.08 & 1.14 & 1.08\\
\hline
100 & 2 & 0.96 & 2690.96 & 551.15 & 1.03 & 2862.72 & 586.33\\
100 & 4 & 0.96 & 1094.71 & 11.20 & 1.03 & 1164.59 & 11.91\\
100 & 8 & 0.96 & 33.21 & 0.97 & 1.03 & 35.33 & 1.03\\
100 & 12 & 0.96 & 0.97 & 0.97 & 1.03 & 1.03 & 1.03\\
100 & 16 & 0.96 & 0.97 & 0.97 & 1.03 & 1.03 & 1.03\\
\hline
200 & 2 & 0.99 & 2264.70 & 139.16 & 1.02 & 2334.75 & 143.47\\
200 & 4 & 0.99 & 408.80 & 1.36 & 1.02 & 421.44 & 1.40\\
200 & 8 & 0.99 & 0.99 & 0.99 & 1.02 & 1.02 & 1.03\\
200 & 12 & 0.99 & 0.99 & 0.99 & 1.02 & 1.02 & 1.02\\
200 & 16 & 0.99 & 0.99 & 0.99 & 1.02 & 1.02 & 1.02\\
\bottomrule
\end{tabular}
\end{table}

\section{Application using multi-site clinical data}
\label{sec-application}

A publicly available dataset of COVID-19 testing results from the Children's Hospital of Pennsylvania (CHOP), available in the R package \texttt{medicaldata} \citep{higgins2021medicaldata}, was analyzed in this study. 
The dataset contained 15,524 de-identified patient records from 88 clinics during the early phase of the COVID-19 pandemic in 2020. 
The polymerase chain reaction (PCR) cycle threshold values were used as responses.
To facilitate a comparison with prior work \citep{limpoco2025linear}, an LMM was fit, including sex, age, drive-through testing status, and age-by-sex interaction as fixed effects.

As IPD were available, the model was first fit using pooled observations and the resulting estimates were treated as a benchmark. 
The proposed summary-based procedure was then applied to DP. 
The privacy parameters were set as $\varepsilon_0 \in \{4,8,12,16\}$ with $\delta = 1/N$, based on preliminary simulation experiments. 
For each privacy level, the analysis was repeated 10,000 times using independently generated Gaussian noise, and the $L_2$ privacy cost and SE inflation factor were evaluated.

The $L_2$ privacy cost and SE inflation factor are summarized in Figure~\ref{fig-application}, and the corresponding quantiles for all $\varepsilon_0$ settings are reported in Appendix~D. 
Even under a stringent privacy setting (e.g., $\varepsilon_0 = 4$), which corresponded to a low reconstruction risk in the simulation study, the $L_2$ privacy cost remained below 0.05 and the SE inflation factor did not exceed 1.5. 
These results indicate that, for this multi-site clinical dataset, meaningful privacy guarantees can be achieved with only a modest loss of statistical efficiency.

\begin{figure}[H]
    \centering
    \includegraphics[width=0.9\textwidth]{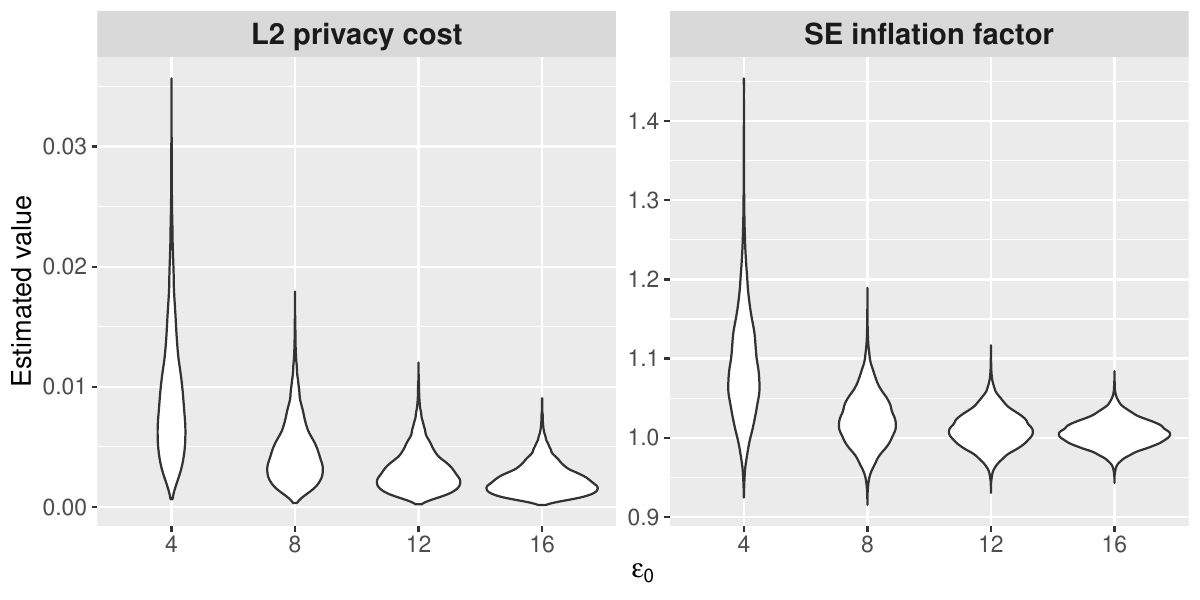}
    \caption{Results of the real data analysis.}
    \label{fig-application}
\end{figure}

\section{Conclusion}
\label{sec-conclusion}

This study developed a one-shot federated inference framework for LMMs, building on the exact likelihood reconstruction from site-level sufficient statistics and incorporating DP through Gaussian perturbation. The results showed that in the absence of privacy protection, sharing site-level quadratic summaries (e.g., Gram-type matrices) can entail a non-negligible reconstruction risk when within-site sample sizes are small and the covariates are discrete. To mitigate this risk while retaining the likelihood-based estimation, a perturbation mechanism was proposed and applied to the sufficient summary statistics required to reconstruct the pooled likelihood. The resulting estimator remained one-shot and communication-efficient, yet provided formal $(\varepsilon,\delta)$-DP protection against reconstruction.

From an inferential perspective, it was established that the proposed approach supports valid multi-site inference without requiring a large within-site sample size. In particular, the consistency and asymptotic normality of the privacy-preserving estimator was proven under a multi-site asymptotic regime, and a closed-form characterization of the leading-order statistical cost of privacy was derived. In addition, cluster-robust (sandwich) variance estimators were introduced that are computable solely from the perturbed summaries. Robust variance estimators enable reliable uncertainty quantification even under covariance mis-specification. Simulation studies confirmed the central trade-off: adding even modest noise substantially reduced reconstruction success, whereas the efficiency loss was governed by $(\varepsilon,\delta)$ and diminished as the number of sites increased. An application to multi-site COVID-19 testing data further illustrated that practical privacy levels can be achieved with limited degradation in point estimation and SEs.

However, several areas merit future research. First, improved privacy accounting and calibration may reduce noise, while maintaining comparable protection. Second, extending the framework beyond the working random-intercept structure would broaden its applicability; however, this may require additional summaries or approximations. Third, developing principled guidance for selecting $(\varepsilon,\delta)$ in clinical multi-site settings, which is potentially informed by reconstruction risk metrics and study-specific utility targets, would further support the deployment. Overall, the proposed method provides a statistically efficient and formally private approach to one-shot federated inference for LMMs in settings in which many sites contribute small samples.

\section*{Supplementary materials}
All the codes related to this paper can be accessed from our Github repository (\url{https://github.com/keisuke-hanada/dp-one-shot-lmm}).
\vspace*{-8pt}

\section*{Acknowledgements}
This work was supported by JSPS KAKENHI Grant Number JP24K23862.
\vspace*{-8pt}

\bibliography{main.bib}
\bibliographystyle{abbrvnat}

\appendix

\section{Formulation of the 0--1 integer programming problem}
\label{app-mip}

Let $\X_k \in \{0,1\}^{n_k \times p}$ be a binary matrix, where the rows correspond to individuals and the columns are binary covariates.
Suppose that only the Gram matrix $(\S_k)_{\X\X} \in \mathbb{Z}_{\ge 0}^{p\times p}$ is observed, where $\mathbb{Z}_{\ge 0}$ denotes the set of nonnegative integers.
Consider the problem of reconstructing $\X_k$ from $(\S_k)_{\X\X}$.

Let $\mathcal{U} \subseteq \{0,1\}^p$ be a set of candidate row patterns. These patterns are collected as rows of a matrix $U \in \{0,1\}^{R\times p}$, where $R = |\mathcal{U}|$ and $U_{r\cdot}$ denotes the $r$-th pattern.
The binary assignment variables are introduced:
\[
  a_{i,r} \in \{0,1\}, \qquad i=1,\dots,n_k,\ \ r=1,\dots,R,
\]
and are interpreted as ``row $i$ of $\X_k$ equals pattern $U_{r\cdot}$.'' Then, $\X_k$ can be represented as follows:
\[
  (\X_k)_{i\cdot} \;=\; \sum_{r=1}^R a_{i,r}\, U_{r\cdot}, \qquad i=1,\dots,n_k.
\]

The reconstruction problem can be formulated as the following $0$--$1$ mixed-integer linear feasibility problem:
\begin{align}
\label{eq:recon-vars}
  \text{find } a_{i,r} \in \{0,1\} ~~\text{subject to } \quad
  &\sum_{r=1}^R a_{i,r} = 1, \quad i = 1,\dots,n_k, \\
  &\sum_{i=1}^{n_k} \sum_{r=1}^R a_{i,r}\, U_{rj} U_{rk}
  = \bigl((\S_k)_{\X\X}\bigr)_{jk}, \quad 1 \le j \le k \le p. \nonumber
\end{align}
No objective function is required (pure feasibility); for example, any mixed-integer programming solver may minimize $0$.

\paragraph{Example: Reconstruction risk in multi-site PCR data.}

An individual-level dataset containing PCR test results and covariates is publicly available from the Children's Hospital of Philadelphia (CHOP) COVID-19 testing data.
Among the 70 sites in this dataset, 25\% had sample size of at most four.
As a concrete illustration, consider the site \texttt{cardiology}, which contained only $n_k=3$ individuals.

The Gram matrix $\X^\top \X$ was constructed from three binary variables: test result (positive/negative), sex (male/female), and whether the specimen was collected via drive-through (yes/no), and then $\X^\top \X$ was shared with the central server.
At this site, the shared Gram matrix was:
\begin{align*}
\bordermatrix{
      & \text{Result} & \text{Sex} & \text{Drive-thru} \cr
\text{Result}    & 1 & 0 & 0 \cr
\text{Sex}    & 0 & 1 & 0 \cr
\text{Drive-thru}& 0 & 0 & 0 \cr
}
\end{align*}

From this matrix, the original covariate matrix can be uniqualy constructed by solving a $0$--$1$ integer programming problem \eqref{eq:recon-vars}.
Indeed, the diagonal entries imply that there was exactly one positive test result, exactly one male individual, and no drive-through specimens.
Using these constraints, an exhaustive search showed that the covariate matrix is uniquely determined up to a permutation of rows, yielding
\begin{align*}
\bordermatrix{
      & \text{Result} & \text{Sex} & \text{Drive-thru} \cr
      & 0 & 0 & 0 \cr
      & 1 & 0 & 0 \cr
      & 0 & 1 & 0 \cr
}
\end{align*}
Consequently, the shared Gram matrix revealed that no male tested positive, whereas exactly one female tested positive.
This example demonstrates that, when within-site sample sizes are small, sharing only quadratic summaries may still pose a non-negligible disclosure risk.
Therefore, methods that (i) share summaries in a form that prevents reliable reconstruction and (ii) still enable statistically valid inferences are needed.

\section{Proofs}
\label{app-proofs}

\subsection{Proof of Theorem \ref{th-consistent-non-dp}}

\begin{proof}[Proof of Theorem \ref{th-consistent-non-dp}]
The ML log-likelihood can be written as follows:
\begin{align*}
    l_{ML}(\bbe, \Psi)
    &= -\frac{1}{2} \left\{ \log |\Sigma| +  (\y - \X \bbe)^\top \Sigma^{-1} (\y - \X \bbe) \right\}.
\end{align*}
where,
\begin{align*}
    \log |\Sigma|
    &= \sum_{k=1}^K \log (\sigma^2)^{n_k-1} (\sigma^2 + n_k \tau^2).
\end{align*}
Focusing on the quadratic term yields:
\begin{align*}
    (\y - \X \bbe)^\top \Sigma^{-1} (\y - \X \bbe)
    &= \sum_{k=1}^K (\y_k - \X_k \bbe)^\top \Sigma_k^{-1} (\y_k - \X_k \bbe).
\end{align*}
Moreover,
\begin{align*}
    &(\y_k - \X_k \bbe)^\top \Sigma_k^{-1} (\y_k - \X_k \bbe) \\
    &= \frac{1}{\sigma^2} \left\{ \y_k^\top \y_k - \y_k^\top \X_k \bbe - \bbe^\top \X_k^\top \y_k + \bbe^\top \X_k^\top \X_k \bbe \right\} \\
    &\quad - \frac{\tau^2}{\sigma^2(\sigma^2+n_k \tau^2)} \left\{\y_k^\top \mathbf{1}_{n_k} \mathbf{1}_{n_k}^\top \y_k
    - \y_k^\top \mathbf{1}_{n_k} \mathbf{1}_{n_k}^\top \X_k \bbe
    - \bbe^\top \X_k^\top \mathbf{1}_{n_k} \mathbf{1}_{n_k}^\top \y_k
    + \bbe^\top \X_k^\top \mathbf{1}_{n_k} \mathbf{1}_{n_k}^\top \X_k \bbe \right\} \\
    &= \frac{1}{\sigma^2} \left\{ (\S_k)_{\y\y} - (\S_k)_{\y\X} \bbe - \bbe^\top (\S_k)_{\X\y} + \bbe^\top (\S_k)_{\X\X} \bbe \right\} \\
    &\quad - \frac{\tau^2}{\sigma^2(\sigma^2+n_k \tau^2)} \left\{(\T_k)_{\y\y} - (\T_k)_{\y\X} \bbe - \bbe^\top (\T_k)_{\X\y} + \bbe^\top (\T_k)_{\X\X} \bbe \right\} \\
    &= \frac{1}{\sigma^2}
    \begin{pmatrix}
        1 \\
        -\bbe
    \end{pmatrix}^\top
    \left\{ \S_k - \frac{\tau^2}{\sigma^2 + n_k \tau^2} \T_k \right\}
    \begin{pmatrix}
        1 \\
        -\bbe
    \end{pmatrix}.
\end{align*}
Therefore,
\begin{align*}
    l_{ML}(\bbe, \sigma^2, \tau^2)
    &= - \frac{1}{2} \sum_{k=1}^K \left[
    \log \bigl\{ (\sigma^2)^{n_k-1}(\sigma^2+n_k \tau^2) \bigr\}
    + \frac{1}{\sigma^2}
    \begin{pmatrix}
        1 \\
        -\bbe
    \end{pmatrix}^\top
    \left\{ \S_k - \frac{\tau^2}{\sigma^2 + n_k \tau^2} \T_k \right\}
    \begin{pmatrix}
        1 \\
        -\bbe
    \end{pmatrix}
    \right].
\end{align*}

The REML log-likelihood is given by:
\begin{align*}
    l_{REML} (\bbe, \sigma^2, \tau^2)
    &= l_{ML}(\bbe, \sigma^2, \tau^2) - \frac{1}{2} \log |\X^\top \Sigma^{-1} \X| \\
    &= l_{ML}(\bbe, \sigma^2, \tau^2) - \frac{1}{2} \log \left|\sum_{k=1}^K \X_k^\top \Sigma_k^{-1} \X_k\right| \\
    &= l_{ML}(\bbe, \sigma^2, \tau^2) - \frac{1}{2} \log \left|\sum_{k=1}^K \frac{1}{\sigma^2}\left\{ (\S_k)_{\X \X}
    - \frac{\tau^2}{\sigma^2+n_k \tau^2} (\T_k)_{\X \X} \right\}\right|.
\end{align*}
Hence, the REML log-likelihood can be expressed solely in terms of the summaries $\{\S_k, \T_k\}_{k=1}^K$.

Finally, the cluster-robust variance estimator $\hat{\V}_{\mathrm{CR0}}$ is expressed using only these summaries. Recall that:
\begin{align*}
    \hat{\V}_{\mathrm{CR0}}
    &= (\X^\top \Sigma^{-1} \X)^{-1} \left( \sum_{k=1}^K \X_k^\top \Sigma_k^{-1} \hat{\e}_k \hat{\e}_k^\top \Sigma_k^{-1} \X_k \right) (\X^\top \Sigma^{-1} \X)^{-1},
\end{align*}
where $\hat{\e}_k=\y_k - \X_k \hat{\bbe}$ \citep{liang1986longitudinal}.
The following is defined:
\begin{align*}
    \W_k &= \X_k^\top \Sigma_k^{-1} \X_k
    = \frac{1}{\sigma^2} (\S_k)_{\X \X}
    - \frac{\tau^2}{\sigma^2(\sigma^2+n_k \tau^2)} (\T_k)_{\X \X}, \\
    \Q_k &= \X_k^\top \Sigma_k^{-1} \y_k
    = \frac{1}{\sigma^2} (\S_k)_{\X \y}
    - \frac{\tau^2}{\sigma^2(\sigma^2+n_k \tau^2)} (\T_k)_{\X \y}, \\
    \P_k
    &= (\Q_k - \W_k \hat{\bbe})
       (\Q_k - \W_k \hat{\bbe})^\top.
\end{align*}
Then, $(\X^\top \Sigma^{-1} \X)^{-1} = \left(\sum_{k=1}^K \W_k\right)^{-1}$, and
\begin{align*}
    \X_k^\top \Sigma_k^{-1} \hat{\e}_k
    &= \X_k^\top \Sigma_k^{-1} (\y_k - \X_k \hat{\bbe})
    = \Q_k - \W_k \hat{\bbe},
\end{align*}
so that
\begin{align*}
    \X_k^\top \Sigma_k^{-1} \hat{\e}_k \hat{\e}_k^\top \Sigma_k^{-1} \X_k
    &= (\Q_k - \W_k \hat{\bbe}) (\Q_k - \W_k \hat{\bbe})^\top
    = \P_k.
\end{align*}
Therefore,
\begin{align*}
    \hat{\V}_{\mathrm{CR0}}
    &= \left(\sum_{k=1}^K \W_k\right)^{-1}
       \left( \sum_{k=1}^K \P_k \right)
       \left(\sum_{k=1}^K \W_k\right)^{-1},
\end{align*}
which completes the proof.
\end{proof}

\subsection{Proof of Proposition \ref{th-gaussian-mechanism-dp}}

\begin{lemma}[Gaussian mechanism; Theorem 1 in \cite{balle2018improving}]
\label{lem-gaussian-dp}
Let $f: \mathcal{X} \to \mathbb{R}^d$ be a deterministic query with global
$L_2$ sensitivity:
\[
\Delta_2 = \sup_{||x - x'||_2 \le r} \| f(x) - f(x') \|_2, \quad r \in (0,\infty).
\]
Then, the mechanism $\mathcal{M}(x)=f(x)+Z$,
$Z\sim N(0,\sigma^2 I_d)$,
is $(\varepsilon,\delta)$-differentially private for
$\sigma = \Delta_2 \sqrt{2\log(1.25/\delta)} / \varepsilon$.
\end{lemma}

\begin{proof}[Proof of Proposition~\ref{th-gaussian-mechanism-dp}]
Let $Q : \mathcal{X} \to \mathbb{R}^{m\times n}$ be a matrix-valued query with Frobenius global sensitivity:
\[
\Delta_F = \sup_{\text{adjacent }\D,\D'} \|Q(\D)-Q(\D')\|_F .
\]
As $\mathbb{R}^{m\times n}$ is isomorphic to $\mathbb{R}^{mn}$ via the vectorization operator $\mathrm{vec}(\cdot)$, it follows that:
\[
\|\A\|_F = \|\mathrm{vec}(\A)\|_2
\]
for any matrix $\A$.
The vector-valued query is defined as $f(\D)=\mathrm{vec}(Q(\D)) \in \mathbb{R}^{mn}$.
Then, its $L_2$ global sensitivity satisfies
\[
\Delta_2
= \sup_{\text{adjacent }\D,\D'} \|f(\D)-f(\D')\|_2
= \sup_{\text{adjacent }\D,\D'} \|Q(\D)-Q(\D')\|_F
= \Delta_F .
\]
Applying Lemma~\ref{lem-gaussian-dp} to $f$ implies that the mechanism:
\[
\mathcal M_g(\D)
= Q(\D)+\U,
\quad
U_{ij}\sim \mathcal N(0,\sigma_{DP}^2),
\]
with
\[
\sigma_{DP}
= \frac{\Delta_F\sqrt{2\log(1.25/\delta)}}{\varepsilon}
\]
is $(\varepsilon,\delta)$-differentially private.
\end{proof}

\subsection{Proof of Theorem \ref{th-consistent-dp}}

The proof relies on the consistency of M-estimators.

\begin{lemma}[Theorem 5.7 in \cite{van2000asymptotic}]
\label{lem-consistency-M-estimator}
Let $M_n$ be random functions and let $M$ be a fixed function of $\bth$ such that for every $\varepsilon > 0$:
\begin{align}
\label{eq-sup-Mn}
\sup_{\bth \in \Theta} \left| M_n(\bth) - M(\bth) \right| \xrightarrow{P} 0, \\
\label{eq-Mth-lt-Mth0}
\sup_{\bth: d(\bth, \bth_0) \ge \varepsilon} M(\bth) < M(\bth_0).
\end{align}
Then, any sequence of estimators $\hat{\bth}_n$ that satisfies $M_n(\hat{\bth}_n) \ge M_n(\bth_0) - o_P(1)$ converges in probability to $\bth_0$.
\end{lemma}

\begin{proof}[Proof of Theorem \ref{th-consistent-dp}]
By Lemma~\ref{lem-consistency-M-estimator}, it suffices to verify \eqref{eq-sup-Mn} and \eqref{eq-Mth-lt-Mth0} for $M_K = K^{-1} l_{DP}$.
Under assumptions (A1) and (A2), \eqref{eq-Mth-lt-Mth0} holds immediately.
Thus, \eqref{eq-sup-Mn} is proved.

Let $M_K^{ML}(\bth) = K^{-1} l_{ML}(\bbe, \sigma^2, \tau^2; \{\S_k, \T_k\}_{k=1}^K)$.
Then
\begin{align}
\label{eq-sup-ineq-Mk}
    \sup_{\bth \in \Theta} \left| M_K(\bth) - M(\bth) \right|
    &\le \sup_{\bth \in \Theta} \left| M_K(\bth) - M_K^{ML}(\bth) \right|
    + \sup_{\bth \in \Theta} \left| M_K^{ML}(\bth) - M(\bth) \right|.
\end{align}
By assumption (A3), the second term on the right-hand side of \eqref{eq-sup-ineq-Mk} converges to zero in probability.
Therefore, the first term is bound.

For any $\bth=(\bbe, \sigma^2, \tau^2) \in \Theta$,
\begin{align*}
    M_K(\bth)
    &= \frac{1}{K} l_{DP}(\bbe, \sigma^2, \tau^2; \{\tilde{\S}_k, \tilde{\T}_k\}_{k=1}^K ) \\
    &= - \frac{1}{2K} \sum_{k=1}^K \left[
    \log (\sigma^2)^{n_k-1}(\sigma^2+n_k \tau^2)
    + \frac{1}{\sigma^2}
    \begin{pmatrix}
        1 \\
        -\bbe
    \end{pmatrix}^\top
    \left\{ \tilde{\S}_k - \frac{\tau^2}{\sigma^2 + n_k \tau^2} \tilde{\T}_k \right\}
    \begin{pmatrix}
        1 \\
        -\bbe
    \end{pmatrix}
    \right] \\
    &= \frac{1}{K} l_{ML}(\bbe, \sigma^2, \tau^2; \{\S_k, \T_k\}_{k=1}^K)
    - \frac{1}{2\sigma^2K} \sum_{k=1}^K
    \left[
    \begin{pmatrix}
        1 \\
        -\bbe
    \end{pmatrix}^\top
    \left\{ \U_{k1} - \frac{\tau^2}{\sigma^2 + n_k \tau^2} \U_{k2} \right\}
    \begin{pmatrix}
        1 \\
        -\bbe
    \end{pmatrix}
    \right] \\
    &= \frac{1}{K} l_{ML}(\bbe, \sigma^2, \tau^2; \{\S_k, \T_k\}_{k=1}^K)
    - \frac{1}{2\sigma^2}
    \begin{pmatrix}
        1 \\
        -\bbe
    \end{pmatrix}^\top
    \left[
    \frac{1}{K} \sum_{k=1}^K \left\{ \U_{k1} - \frac{\tau^2}{\sigma^2 + n_k \tau^2} \U_{k2} \right\}
    \right]
    \begin{pmatrix}
        1 \\
        -\bbe
    \end{pmatrix}.
\end{align*}
Let $\v_\bbe = (1~-\bbe^\top)^\top$ and $\A_k = \U_{k1} - \frac{\tau^2}{\sigma^2 + n_k \tau^2} \U_{k2}$.
Then,
\begin{align*}
    \sup_{\bth \in \Theta} \left| M_K(\bth) - M_K^{ML}(\bth) \right|
    &=
    \sup_{\bth \in \Theta} \left| \frac{1}{2\sigma^2K} \sum_{k=1}^K \v_\bbe^\top \A_k \v_\bbe \right| \\
    &\le
    \frac{1}{2\underline{\sigma}^2}
    \sum_{s=1}^{p+1} \sum_{r=1}^{p+1}
    \sup_{\bth \in \Theta}
    \left| \frac{1}{K} \sum_{k=1}^K (\v_\bbe)_s (\v_\bbe)_r (\A_k)_{rs} \right|.
\end{align*}
Owing to the compactness of $\Theta$ (A1), there exist constants $C_1,C_2<\infty$ such that
$|(\v_\bbe)_s (\v_\bbe)_r| \le C_1$ for $s,r=1,\dots,p+1$, and
$\tau^2/(\sigma^2+n_k \tau^2) \le \overline{\tau}^2/\underline{\sigma}^2 = C_2$ for $k=1,\dots,K$.
By assumption (A4), each element of $\U_{k1}$ and $\U_{k2}$ has mean zero and finite variance $\sigma_{DP}^2$. Hence, by the weak law of large numbers,
$\frac{1}{K} \sum_{k=1}^K (\U_{k\ell})_{rs} \xrightarrow{P} 0$ for $\ell=1,2$ as $K \to \infty$.
Therefore,
\begin{align*}
    \sup_{\bth \in \Theta} \left| M_K(\bth) - M_K^{ML}(\bth) \right|
    &\le
    \frac{1}{2\underline{\sigma}^2}
    \sum_{s=1}^{p+1} \sum_{r=1}^{p+1}
    C_1
    \left\{
    \left| \frac{1}{K} \sum_{k=1}^K (\U_{k1})_{rs} \right|
    + C_2 \left| \frac{1}{K} \sum_{k=1}^K (\U_{k2})_{rs} \right|
    \right\}
    \xrightarrow{P} 0.
\end{align*}
Consequently,
\[
\sup_{\bth \in \Theta} \left| M_K(\bth) - M(\bth) \right| \xrightarrow{P} 0.
\]

Finally, the estimator:
\[
\hat{\bth}_{DP} = \arg\max_{\bth \in \Theta} l_{DP}(\bth)
\]
satisfies $M_K(\hat{\bth}_{DP}) \ge M_K(\bth_0) \ge M_K(\bth_0)-o_P(1)$.
Thus, Lemma~\ref{lem-consistency-M-estimator} yields $\hat{\bth}_{DP} \xrightarrow{P} \bth_0$ as $K \to \infty$.
\end{proof}

\subsection{Proof of Theorem \ref{th-asymp-dp}}

The proof relies on the following lemma.

\begin{lemma}[Theorem 5.23 in \cite{van2000asymptotic}]
\label{lem-asymp-normal-m-est}
For each $\theta$ in an open subset of Euclidean space, let $x \mapsto m_{\theta}(x)$ be a measurable function such that $\theta \mapsto m_{\theta}(x)$ is differentiable at $\theta_0$ for $P$-almost every $x$ with derivative $\dot{m}_{\theta_0}(x)$ and such that, for every $\theta_1$ and $\theta_2$ in a neighborhood of $\theta_0$ and a measurable function $\dot{m}$ with $P \dot{m} < \infty$:
\[
\left| m_{\theta_1}(x) - m_{\theta_2}(x) \right| \le \dot{m}(x) \| \theta_1 - \theta_2 \|.
\]
Furthermore, assume that the map $\theta \mapsto Pm_{\theta}$ admits a second-order Taylor expansion at the point of maximum $\theta_0$ with a non-singular symmetric second-derivative matrix $V_{\theta_0}$. If $\mathbb{P}_n m_{\hat{\theta}_n} \ge \sup_{\theta} \mathbb{P}_n m_{\theta} - o_P(n^{-1})$ and $\hat{\theta}_n \xrightarrow{P} \theta_0$, then
\[
\sqrt{n} (\hat{\theta}_n - \theta_0) = - V_{\theta_0}^{-1} \frac{1}{\sqrt{n}} \sum_{i=1}^n \dot{m}_{\theta_0} (X_i) + o_P(1).
\]
In particular, $\sqrt{n} (\hat{\theta}_n - \theta_0)$ is asymptotically normal with mean zero and covariance matrix $V_{\theta_0}^{-1} P \dot{m}_{\theta_0} \dot{m}_{\theta_0}^\top V_{\theta_0}^{-1}$.
\end{lemma}

\begin{proof}[Proof of Theorem \ref{th-asymp-dp}]
It is verified that Lemma~\ref{lem-asymp-normal-m-est} applies to the estimator sequence $\hat{\bth}_{DP}$.

By the compactness assumption (A1) and the analyticity of $m_{\bth}(\mathcal{D}_k)$, there exists a neighborhood $U(\bth_0) \subset \Theta$ such that
\[
\sup_{\bth \in U(\bth_0)} \|\dot{m}_{\bth}(\mathcal{D}_k)\| < \infty.
\]
Furthermore, assumption (A5)(i) implies that $P\|\dot m_{\bth_0}(\mathcal{D}_k)\|^2 < \infty$, hence
$\dot m(\mathcal{D}_k) = \sup_{\bth \in U(\bth_0)}\|\dot m_\bth(\mathcal D_k)\|$
is finite and $P$-integrable.
Using the mean-value theorem, for every $\bth_1,\bth_2 \in U(\bth_0)$:
\[
|m_{\bth_1}(\mathcal{D}_k) - m_{\bth_2}(\mathcal{D}_k)|
\le \dot{m}(\mathcal{D}_k)\, \|\bth_1 - \bth_2\|.
\]

By the identifiability assumption (A2), $\bth_0$ is a unique maximizer of $M(\bth)=P[m_{\bth}(\mathcal{D}_k)]$.
As $m_\bth(\mathcal{D}_k)$ is twice differentiable, $P m_\bth$ admits a second-order Taylor expansion around $\bth_0$.
Assumption (A5)(ii) ensures that the expected Hessian matrix
\[
V_{\bth_0} = P\left[ \ddot{m}_{\bth_0}(\mathcal{D}_k) \right]
\]
is finite and nonsingular.

According to the definition of the ML estimator:
\[
l_{DP}(\hat{\bth}_{DP}) \ge \sup_{\bth \in \Theta} l_{DP}(\bth) - o_P(1).
\]
Dividing by $K$ yields
\[
M_K(\hat{\bth}_{DP}) = K^{-1} l_{DP}(\hat{\bth}_{DP})
\ge \sup_{\bth \in \Theta} M_{K}(\bth) - o_P(K^{-1}).
\]
In addition, Theorem~\ref{th-consistent-dp} implies that $\hat{\bth}_{DP} \xrightarrow{P} \bth_0$.

Therefore, Lemma~\ref{lem-asymp-normal-m-est} yields:
\[
\sqrt{K} (\hat{\bth}_{DP} - \bth_0)
= - V_{\bth_0}^{-1} \frac{1}{\sqrt{K}} \sum_{k=1}^K \dot{m}_{\bth_0}(\mathcal{D}_k) + o_P(1).
\]
Moreover, $\sqrt{K} (\hat{\bth}_{DP} - \bth_0)$ is asymptotically normal with mean zero and covariance matrix $\V_{DP}$.

The consistency of the sandwich variance estimator $\hat{\V}_{DP}$ follows from the continuous mapping theorem. Specifically, the following are defined:
\begin{align*}
    \hat{\H}_K &= \frac{1}{K} \sum_{k=1}^K \tilde{\W}_k, \\
    \hat{\B}_K &= \frac{1}{K^2} \sum_{k=1}^K \left\{
    \tilde{\Q}_k \tilde{\Q}_k^\top
    - \tilde{\Q}_k \hat{\bbe}^\top \tilde{\W}_k
    - \tilde{\W}_k \hat{\bbe} \tilde{\Q}_k^\top
    + \tilde{\W}_k \hat{\bbe} \hat{\bbe}^\top \tilde{\W}_k
    \right\}.
\end{align*}
Then, $\hat{\H}_K \xrightarrow{P} (P[\ddot{m}_{\bth_0}(\mathcal{D}_k)])_{\bbe\bbe}$ and
$\hat{\B}_K \xrightarrow{P} (P[\dot{m}_{\bth_0}(\mathcal{D}_k)\dot{m}_{\bth_0}(\mathcal{D}_k)^\top])_{\bbe\bbe}$, and hence:
\[
\hat{\V}_{DP} = \hat{\H}_K^{-1} \hat{\B}_K \hat{\H}_K^{-1}
\xrightarrow{P} (\V_{DP})_{\bbe\bbe}.
\]
\end{proof}

\subsection{Proof of Proposition \ref{th-dp-loss}}

\begin{proof}[Proof of Proposition \ref{th-dp-loss}]
Let $\Delta \W_k = \tilde{\W}_k - \W_k$ and $\Delta \Q_k = \tilde{\Q}_k - \Q_k$. Then,
\begin{align*}
    \Delta \W_k &= \frac{1}{\sigma^2} (\U_{k1})_{\X \X}
    - \frac{\tau^2}{\sigma^2(\sigma^2+n_k \tau^2)} (\U_{k2})_{\X \X}, \\
    \Delta \Q_k &= \frac{1}{\sigma^2} (\U_{k1})_{\X \y}
    - \frac{\tau^2}{\sigma^2(\sigma^2+n_k \tau^2)} (\U_{k2})_{\X \y}.
\end{align*}
Using a Neumann series expansion:
\begin{align*}
    \tilde{\W}^{-1}
    &= \left( \W + \Delta \W \right)^{-1}
     = \W^{-1} - \W^{-1} \Delta \W \W^{-1} + R_{W,K}, \qquad \|R_{W,K}\| = o_p(K^{-1}),
\end{align*}
where $\Delta \W = \sum_{k=1}^K \Delta \W_k$.
Therefore,
\begin{align*}
    \hat{\bbe}_{DP} - \hat{\bbe}
    &= \W^{-1} \Delta \Q - \W^{-1} \Delta \W \W^{-1} \tilde{\Q} + O(\|\Delta \W\|_F^2),
\end{align*}
where $\Delta \Q = \sum_{k=1}^K \Delta \Q_k$.
Let $r_K=(\hat{\bbe}_{DP}-\hat{\bbe}) - \W^{-1}\Delta \Q$, such that $\hat{\bbe}_{DP} - \hat{\bbe} = \W^{-1} \Delta \Q + r_K$.
Under assumptions (A1)--(A4):
\[
\|r_K\| = o_p(K^{-1/2}) \quad\text{and}\quad E\|r_K\|^2 = o(K^{-1}).
\]
This follows from $\|\Delta\W\|_F = O_p(\sqrt{K})$, $\|\W^{-1}\| = O(K^{-1})$, and $\|\hat{\bbe}_{DP}\| = O_p(1)$ \citep{van2000asymptotic}.

Let $\W_0 = \lim_{K \to \infty} \frac{1}{K} \W$ and $\Sigma_\Q = V[\Delta \Q]/(K\sigma_{DP}^2)$. Then,
\begin{align*}
    E\| \W^{-1} \Delta \Q \|^2
    &= \text{tr}\left( E[\W^{-1} \Delta \Q \Delta \Q^\top \W^{-1}] \right) \\
    &= \text{tr}\left( \W^{-1} E[\Delta \Q \Delta \Q^\top] \W^{-1} \right) \\
    &= \text{tr}\left( \W^{-1} K \sigma_{DP}^2 \Sigma_\Q \W^{-1} \right) \\
    &= \text{tr}(\W_0^{-1} \Sigma_\Q \W_0^{-1}) \frac{\sigma_{DP}^2}{K} + o(K^{-1}),
\end{align*}
where the independence between $\Delta\Q$ and $\W$ is used, and $\W/K \xrightarrow{P} \W_0$.
Hence,
\begin{align*}
    E\|\hat{\bbe}_{DP} - \hat{\bbe}\|^2
    &= E\|\W^{-1} \Delta \Q + r_K \|^2 \\
    &= \text{tr}(\W_0^{-1} \Sigma_\Q \W_0^{-1}) \frac{\sigma_{DP}^2}{K} + o(K^{-1}) \\
    &= \text{tr}(\Omega) \frac{\sigma_{DP}^2}{K} + o(K^{-1}),
\end{align*}
where $\Omega = \W_0^{-1} \Sigma_\Q \W_0^{-1}$.
\end{proof}

\section{Additional simulation results}

\begin{figure}[h]
    \centering
    \includegraphics[width=0.95\textwidth, page=5]{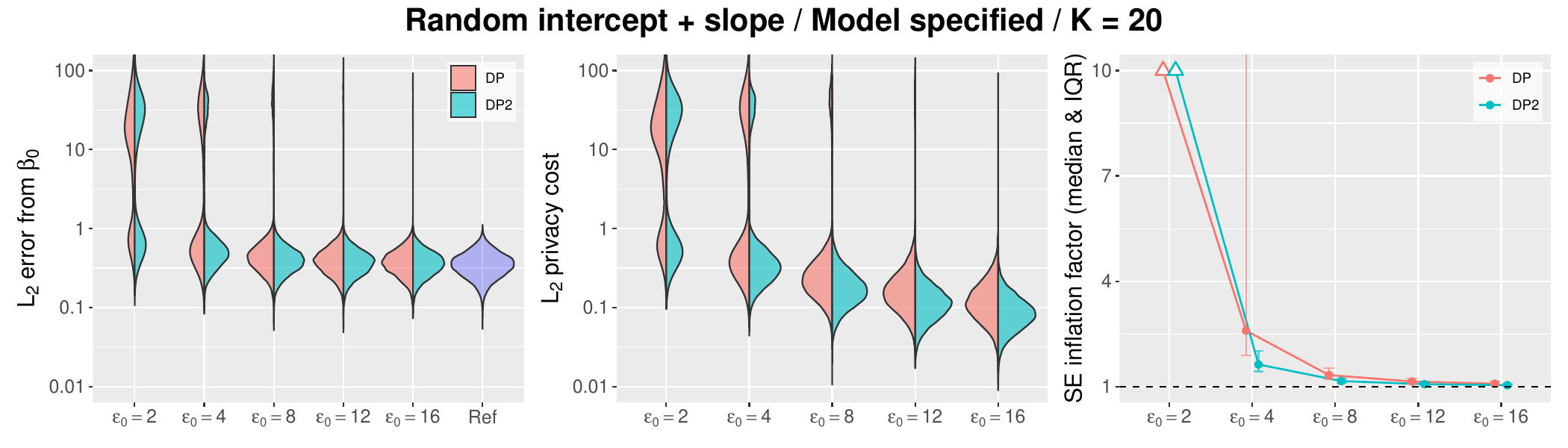}
    \includegraphics[width=0.95\textwidth, page=6]{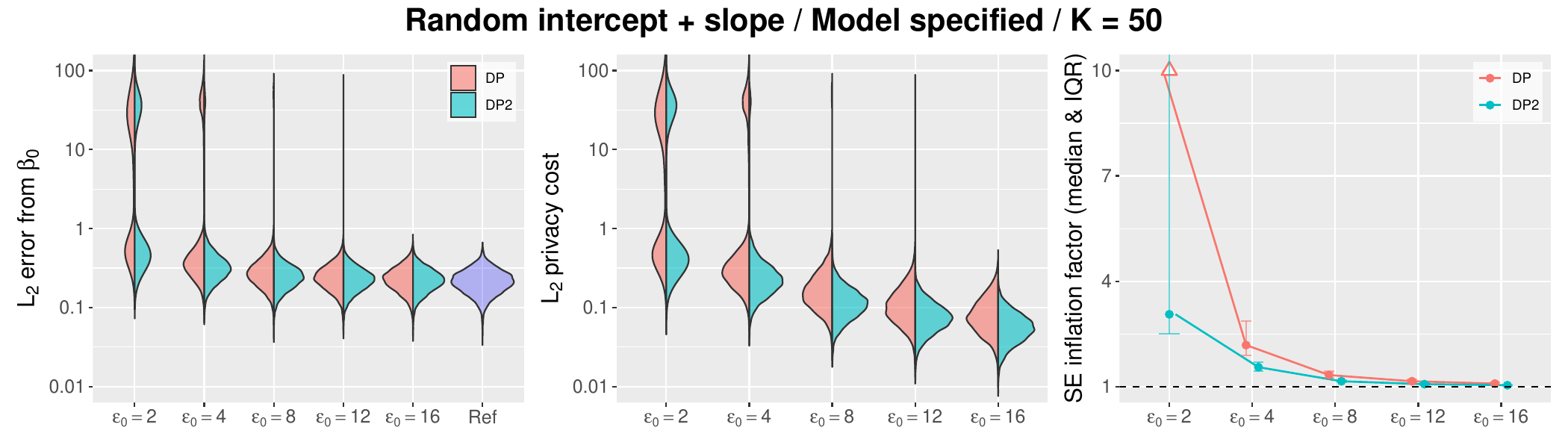}
    \includegraphics[width=0.95\textwidth, page=7]{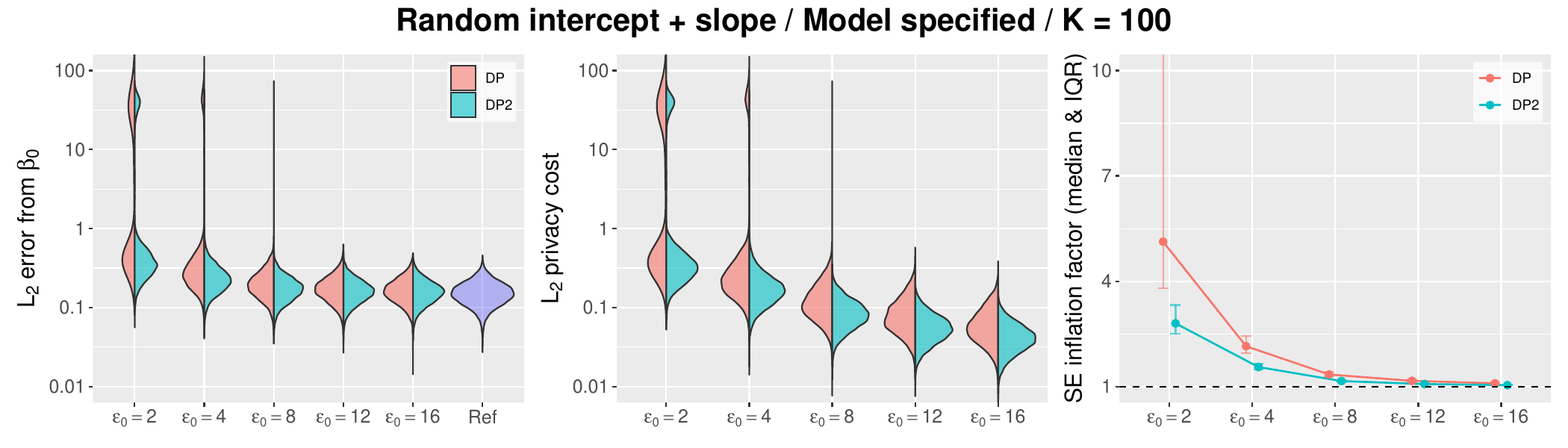}
    \includegraphics[width=0.95\textwidth, page=8]{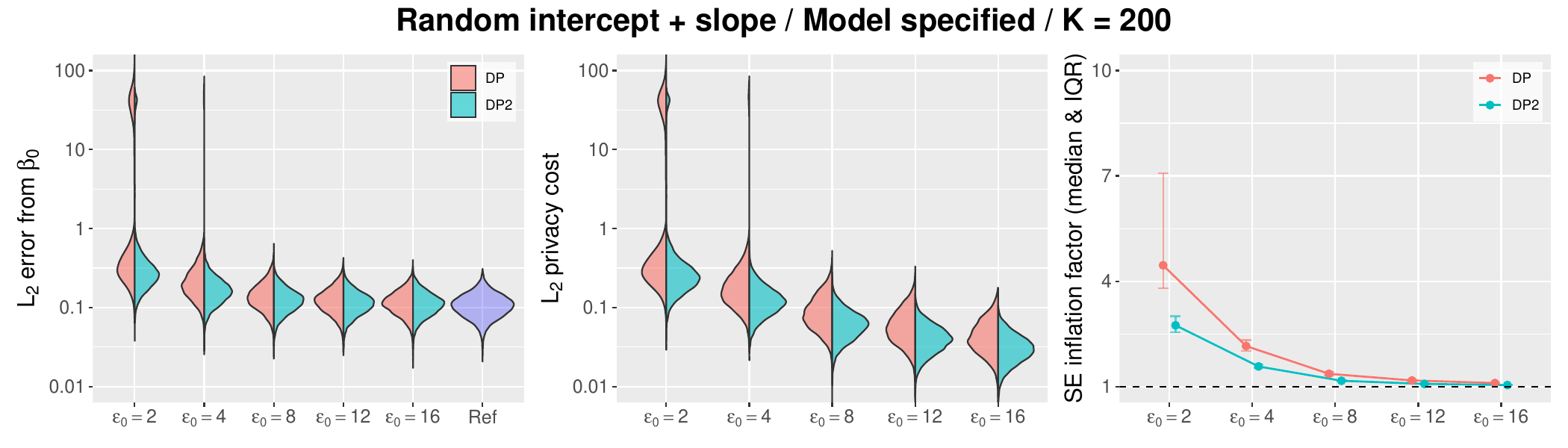}
    \caption{Overview of simulation results under the correctly specified random-intercept and slope model. Triangles indicate median SE inflation factors exceeding the plotting range; IQRs are omitted in those cases.}
    \label{fig-sim-estimation-add1}
\end{figure}

\begin{figure}[H]
    \centering
    \includegraphics[width=0.95\textwidth, page=9]{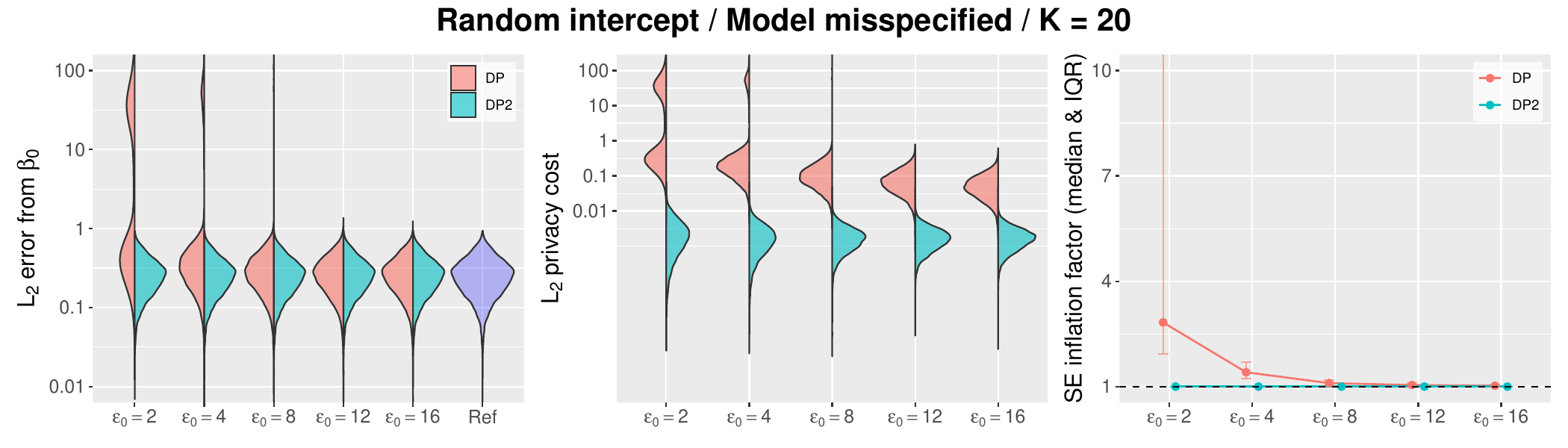}
    \includegraphics[width=0.95\textwidth, page=10]{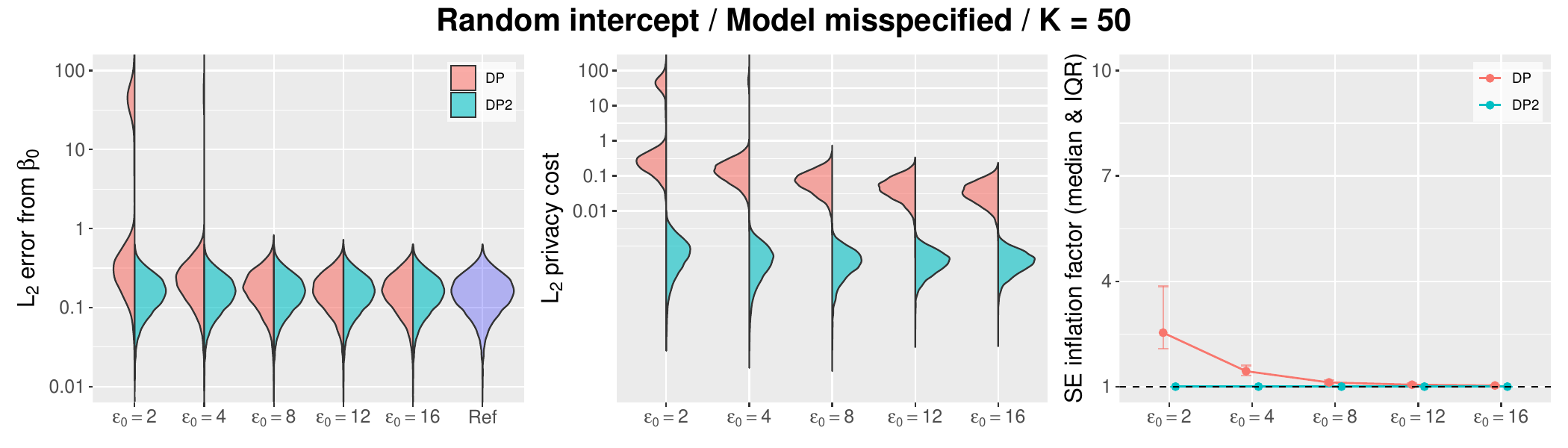}
    \includegraphics[width=0.95\textwidth, page=11]{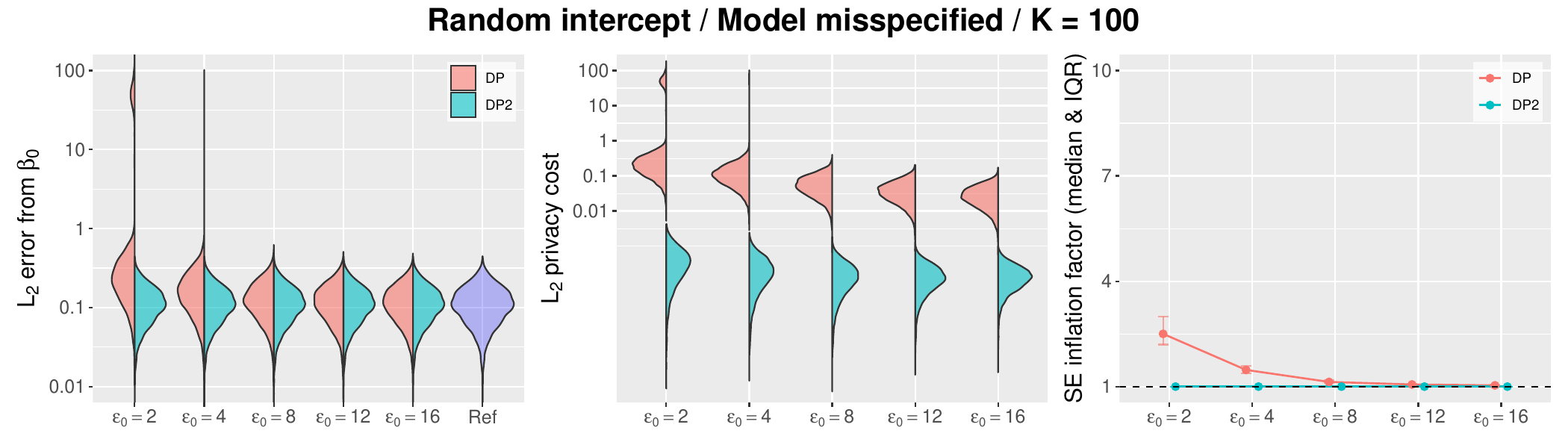}
    \includegraphics[width=0.95\textwidth, page=12]{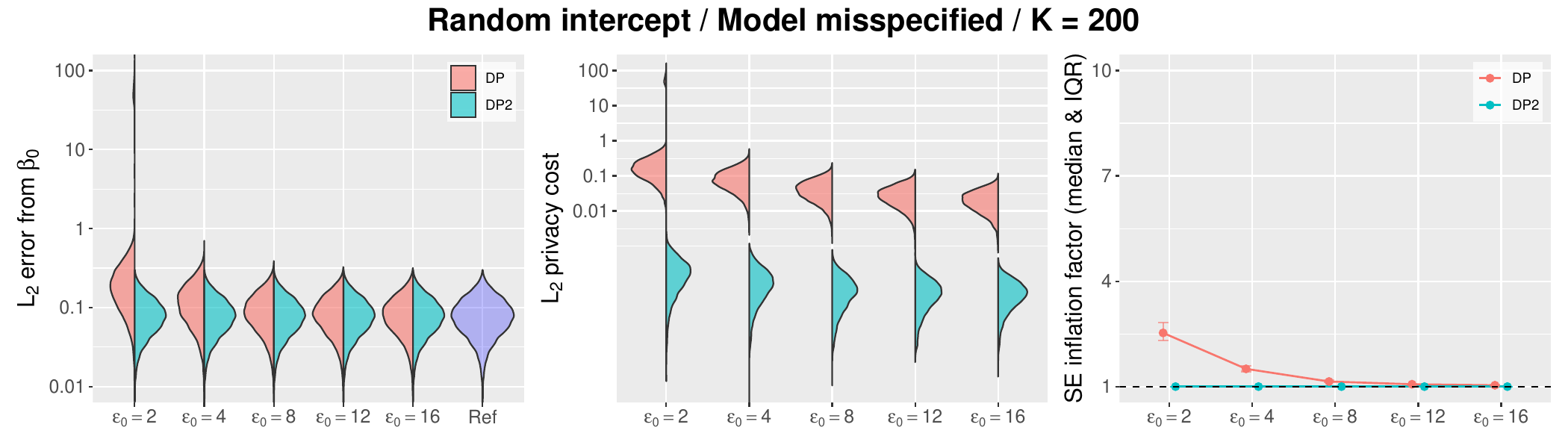}
    \caption{Overview of simulation results under the  the misspecified random-intercept model.}
    \label{fig-sim-estimation-add2}
\end{figure}

\begin{figure}[H]
    \centering
    \includegraphics[width=0.95\textwidth, page=13]{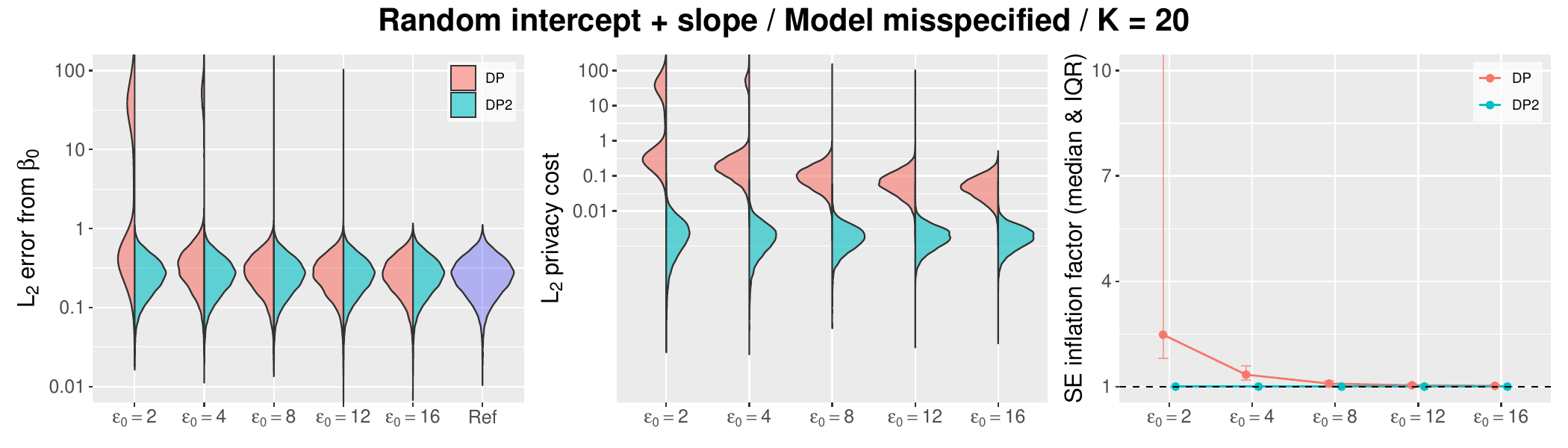}
    \includegraphics[width=0.95\textwidth, page=14]{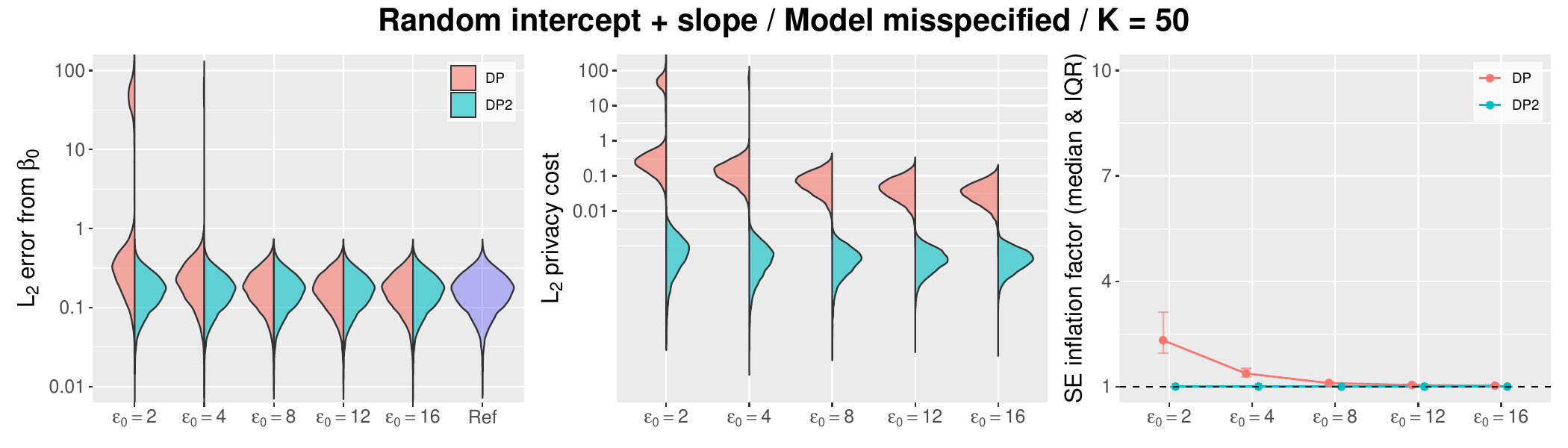}
    \includegraphics[width=0.95\textwidth, page=15]{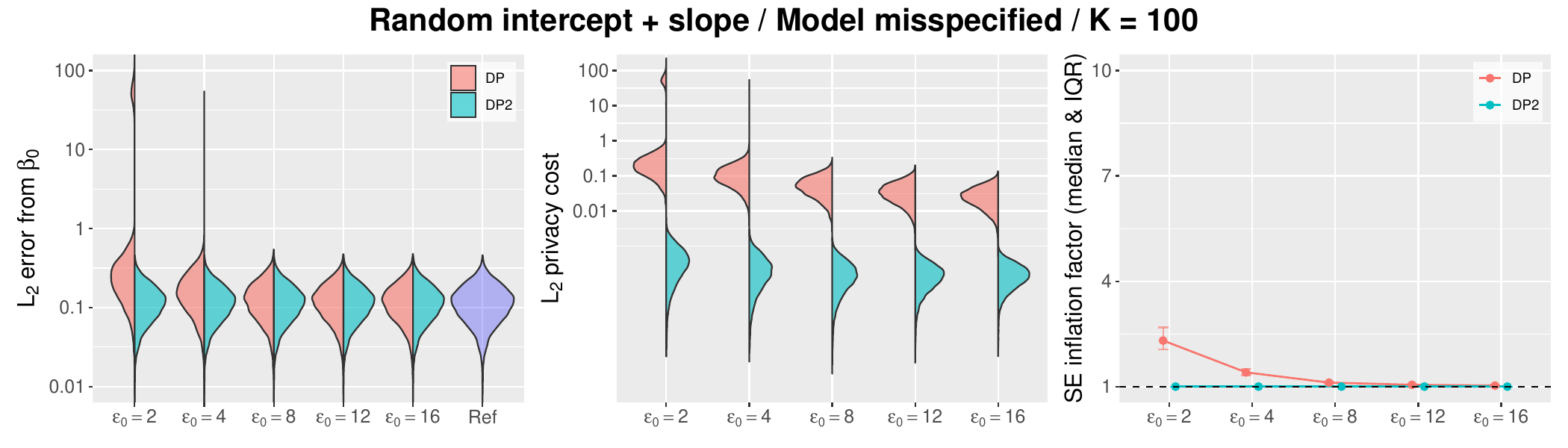}
    \includegraphics[width=0.95\textwidth, page=16]{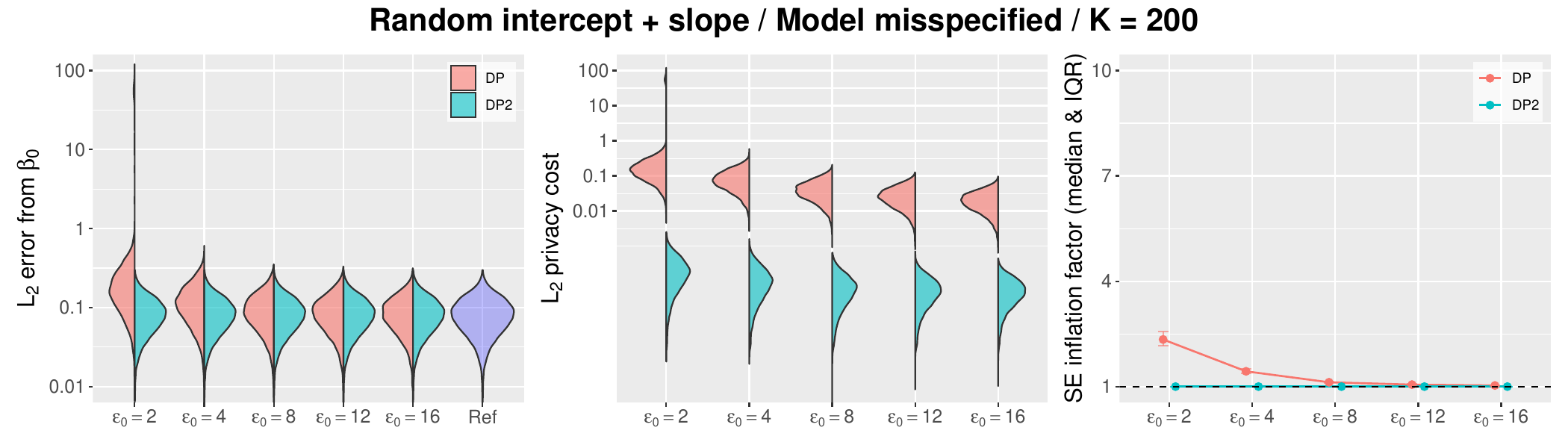}
    \caption{Overview of simulation results under the misspecified random-intercept and slope model.}
    \label{fig-sim-estimation-add3}
\end{figure}

\newpage

\begin{landscape}
\begingroup\fontsize{9}{11}\selectfont

\begin{longtable}[H]{lllc|ccc|ccc|ccc|ccc}
\caption{SE calibration ratio for $x_1$.}\\
\toprule
Model & FM & Strata & $\varepsilon$
& \multicolumn{3}{c}{CR0}
& \multicolumn{3}{c}{CR1}
& \multicolumn{3}{c}{CR1p}
& \multicolumn{3}{c}{CR1S} \\

\cmidrule(lr){5-7} \cmidrule(lr){8-10} \cmidrule(lr){11-13} \cmidrule(lr){14-16}

& & & 
& IPD & DP & DP2
& IPD & DP & DP2
& IPD & DP & DP2
& IPD & DP & DP2 \\

\midrule
\endfirsthead

\caption[]{SE calibration ratio for $x_1$. \textit{(continued)}}\\
\toprule
Model & FM & Strata & $\varepsilon$
& \multicolumn{3}{c}{CR0}
& \multicolumn{3}{c}{CR1}
& \multicolumn{3}{c}{CR1p}
& \multicolumn{3}{c}{CR1S} \\

\cmidrule(lr){5-7} \cmidrule(lr){8-10} \cmidrule(lr){11-13} \cmidrule(lr){14-16}

& & &
& IPD & DP & DP2
& IPD & DP & DP2
& IPD & DP & DP2
& IPD & DP & DP2 \\

\midrule
\endhead

\endfoot
\bottomrule
\endlastfoot

Int & Spec & 20 & 1 & 0.86 & 11938.51 & 9147.77 & 0.90 & 12566.85 & 9629.24 & 1.23 & 17055.01 & 13068.25 & 0.92 & 12727.04 & 9751.98\\
Int & Spec & 20 & 2 & 0.86 & 1745.41 & 4444.54 & 0.90 & 1837.28 & 4678.46 & 1.23 & 2493.45 & 6349.34 & 0.92 & 1860.70 & 4738.10\\
Int & Spec & 20 & 4 & 0.86 & 1348.82 & 1463.55 & 0.90 & 1419.81 & 1540.58 & 1.23 & 1926.88 & 2090.79 & 0.92 & 1437.91 & 1560.22\\
Int & Spec & 20 & 6 & 0.86 & 1182.13 & 364.15 & 0.90 & 1244.35 & 383.31 & 1.23 & 1688.76 & 520.21 & 0.92 & 1260.21 & 388.20\\
Int & Spec & 20 & 8 & 0.86 & 1205.85 & 34.39 & 0.90 & 1269.31 & 36.20 & 1.23 & 1722.64 & 49.13 & 0.92 & 1285.49 & 36.66\\
\hline
Int & Spec & 20 & 10 & 0.86 & 680.71 & 6.87 & 0.90 & 716.54 & 7.23 & 1.23 & 972.44 & 9.81 & 0.92 & 725.67 & 7.32\\
Int & Spec & 20 & 12 & 0.86 & 308.16 & 5.82 & 0.90 & 324.38 & 6.13 & 1.23 & 440.23 & 8.32 & 0.92 & 328.52 & 6.21\\
Int & Spec & 20 & 16 & 0.86 & 49.15 & 0.86 & 0.90 & 51.74 & 0.91 & 1.23 & 70.21 & 1.23 & 0.92 & 52.40 & 0.92\\
Int & Spec & 20 & 20 & 0.86 & 14.12 & 0.86 & 0.90 & 14.86 & 0.91 & 1.23 & 20.17 & 1.23 & 0.92 & 15.05 & 0.92\\
Int & Spec & 20 & Inf & 0.86 & 0.86 & 0.86 & 0.90 & 0.90 & 0.90 & 1.23 & 1.23 & 1.23 & 0.92 & 0.92 & 0.92\\
\hline
Int & Spec & 50 & 1 & 0.95 & 2811.63 & 1474.92 & 0.97 & 2869.01 & 1505.02 & 1.08 & 3195.04 & 1676.04 & 0.97 & 2883.56 & 1512.65\\
Int & Spec & 50 & 2 & 0.95 & 3217.91 & 1724.02 & 0.97 & 3283.58 & 1759.20 & 1.08 & 3656.72 & 1959.11 & 0.97 & 3300.24 & 1768.12\\
Int & Spec & 50 & 4 & 0.95 & 1556.33 & 111.96 & 0.97 & 1588.09 & 114.25 & 1.08 & 1768.56 & 127.23 & 0.97 & 1596.15 & 114.83\\
Int & Spec & 50 & 6 & 0.95 & 806.62 & 6.93 & 0.97 & 823.09 & 7.07 & 1.08 & 916.62 & 7.87 & 0.97 & 827.26 & 7.11\\
Int & Spec & 50 & 8 & 0.95 & 297.79 & 0.98 & 0.97 & 303.86 & 1.00 & 1.08 & 338.39 & 1.11 & 0.97 & 305.41 & 1.00\\
\hline
Int & Spec & 50 & 10 & 0.95 & 49.53 & 0.95 & 0.97 & 50.54 & 0.97 & 1.08 & 56.28 & 1.08 & 0.97 & 50.80 & 0.98\\
Int & Spec & 50 & 12 & 0.95 & 19.61 & 0.95 & 0.97 & 20.01 & 0.97 & 1.08 & 22.29 & 1.08 & 0.97 & 20.11 & 0.98\\
Int & Spec & 50 & 16 & 0.95 & 1.01 & 0.95 & 0.97 & 1.03 & 0.97 & 1.08 & 1.14 & 1.08 & 0.97 & 1.03 & 0.98\\
Int & Spec & 50 & 20 & 0.95 & 0.95 & 0.95 & 0.97 & 0.97 & 0.97 & 1.08 & 1.08 & 1.08 & 0.97 & 0.97 & 0.97\\
Int & Spec & 50 & Inf & 0.95 & 0.95 & 0.95 & 0.97 & 0.97 & 0.97 & 1.08 & 1.08 & 1.08 & 0.97 & 0.97 & 0.97\\
\hline
Int & Spec & 100 & 1 & 0.96 & 2080.02 & 3932.60 & 0.97 & 2101.03 & 3972.32 & 1.03 & 2212.79 & 4183.62 & 0.98 & 2106.36 & 3982.39\\
Int & Spec & 100 & 2 & 0.96 & 2690.96 & 551.15 & 0.97 & 2718.14 & 556.71 & 1.03 & 2862.72 & 586.33 & 0.98 & 2725.03 & 558.12\\
Int & Spec & 100 & 4 & 0.96 & 1094.71 & 11.20 & 0.97 & 1105.77 & 11.31 & 1.03 & 1164.59 & 11.91 & 0.98 & 1108.57 & 11.34\\
Int & Spec & 100 & 6 & 0.96 & 279.24 & 0.97 & 0.97 & 282.06 & 0.98 & 1.03 & 297.06 & 1.03 & 0.98 & 282.77 & 0.98\\
Int & Spec & 100 & 8 & 0.96 & 33.21 & 0.97 & 0.97 & 33.54 & 0.98 & 1.03 & 35.33 & 1.03 & 0.98 & 33.63 & 0.98\\
\hline
Int & Spec & 100 & 10 & 0.96 & 0.97 & 0.97 & 0.97 & 0.98 & 0.98 & 1.03 & 1.03 & 1.03 & 0.98 & 0.98 & 0.98\\
Int & Spec & 100 & 12 & 0.96 & 0.97 & 0.97 & 0.97 & 0.98 & 0.97 & 1.03 & 1.03 & 1.03 & 0.98 & 0.98 & 0.98\\
Int & Spec & 100 & 16 & 0.96 & 0.97 & 0.97 & 0.97 & 0.98 & 0.98 & 1.03 & 1.03 & 1.03 & 0.98 & 0.98 & 0.98\\
Int & Spec & 100 & 20 & 0.96 & 0.97 & 0.97 & 0.97 & 0.98 & 0.98 & 1.03 & 1.03 & 1.03 & 0.98 & 0.98 & 0.98\\
Int & Spec & 100 & Inf & 0.96 & 0.96 & 0.96 & 0.97 & 0.97 & 0.97 & 1.03 & 1.03 & 1.03 & 0.98 & 0.98 & 0.98\\
\hline
Int & Spec & 200 & 1 & 0.99 & 3220.96 & 1735.10 & 1.00 & 3237.14 & 1743.82 & 1.02 & 3320.57 & 1788.76 & 1.00 & 3241.23 & 1746.02\\
Int & Spec & 200 & 2 & 0.99 & 2264.70 & 139.16 & 1.00 & 2276.08 & 139.86 & 1.02 & 2334.75 & 143.47 & 1.00 & 2278.96 & 140.04\\
Int & Spec & 200 & 4 & 0.99 & 408.80 & 1.36 & 1.00 & 410.85 & 1.37 & 1.02 & 421.44 & 1.40 & 1.00 & 411.37 & 1.37\\
Int & Spec & 200 & 6 & 0.99 & 37.07 & 1.00 & 1.00 & 37.26 & 1.01 & 1.02 & 38.22 & 1.03 & 1.00 & 37.30 & 1.01\\
Int & Spec & 200 & 8 & 0.99 & 0.99 & 0.99 & 1.00 & 1.00 & 1.00 & 1.02 & 1.02 & 1.03 & 1.00 & 1.00 & 1.00\\
\hline
Int & Spec & 200 & 10 & 0.99 & 1.00 & 0.99 & 1.00 & 1.00 & 1.00 & 1.02 & 1.03 & 1.02 & 1.00 & 1.00 & 1.00\\
Int & Spec & 200 & 12 & 0.99 & 0.99 & 0.99 & 1.00 & 1.00 & 1.00 & 1.02 & 1.02 & 1.02 & 1.00 & 1.00 & 1.00\\
Int & Spec & 200 & 16 & 0.99 & 0.99 & 0.99 & 1.00 & 1.00 & 1.00 & 1.02 & 1.02 & 1.02 & 1.00 & 1.00 & 1.00\\
Int & Spec & 200 & 20 & 0.99 & 0.99 & 0.99 & 1.00 & 1.00 & 1.00 & 1.02 & 1.02 & 1.03 & 1.00 & 1.00 & 1.00\\
Int & Spec & 200 & Inf & 0.99 & 0.99 & 0.99 & 1.00 & 1.00 & 1.00 & 1.02 & 1.02 & 1.02 & 1.00 & 1.00 & 1.00\\
\hline
Int + slope & Spec & 20 & 1 & 0.84 & 7896.17 & 11659.57 & 0.89 & 8311.75 & 12273.23 & 1.20 & 11280.24 & 16656.53 & 0.90 & 8418.73 & 12431.20\\
Int + slope & Spec & 20 & 2 & 0.84 & 2818.28 & 6795.41 & 0.89 & 2966.61 & 7153.06 & 1.20 & 4026.11 & 9707.73 & 0.90 & 3004.79 & 7245.13\\
Int + slope & Spec & 20 & 4 & 0.84 & 3949.38 & 1294.85 & 0.89 & 4157.24 & 1363.00 & 1.20 & 5641.97 & 1849.79 & 0.90 & 4210.75 & 1380.55\\
Int + slope & Spec & 20 & 6 & 0.84 & 2009.13 & 127.16 & 0.89 & 2114.87 & 133.86 & 1.20 & 2870.18 & 181.66 & 0.90 & 2142.09 & 135.58\\
Int + slope & Spec & 20 & 8 & 0.84 & 841.87 & 20.64 & 0.89 & 886.18 & 21.73 & 1.20 & 1202.67 & 29.49 & 0.90 & 897.59 & 22.01\\
\hline
Int + slope & Spec & 20 & 10 & 0.84 & 438.88 & 1.88 & 0.89 & 461.97 & 1.98 & 1.20 & 626.96 & 2.69 & 0.90 & 467.92 & 2.00\\
Int + slope & Spec & 20 & 12 & 0.84 & 197.99 & 3.71 & 0.89 & 208.41 & 3.91 & 1.20 & 282.84 & 5.31 & 0.90 & 211.09 & 3.96\\
Int + slope & Spec & 20 & 16 & 0.84 & 20.36 & 0.85 & 0.89 & 21.43 & 0.89 & 1.20 & 29.09 & 1.21 & 0.90 & 21.71 & 0.90\\
Int + slope & Spec & 20 & 20 & 0.84 & 2.17 & 0.85 & 0.89 & 2.28 & 0.89 & 1.20 & 3.10 & 1.21 & 0.90 & 2.31 & 0.90\\
Int + slope & Spec & 20 & Inf & 0.84 & 0.84 & 0.84 & 0.89 & 0.89 & 0.89 & 1.20 & 1.20 & 1.20 & 0.90 & 0.90 & 0.90\\
\hline
Int + slope & Spec & 50 & 1 & 0.96 & 4813.74 & 1689.86 & 0.98 & 4911.98 & 1724.35 & 1.10 & 5470.16 & 1920.30 & 0.99 & 4936.98 & 1733.12\\
Int + slope & Spec & 50 & 2 & 0.96 & 4847.80 & 1536.82 & 0.98 & 4946.73 & 1568.19 & 1.10 & 5508.86 & 1746.39 & 0.99 & 4971.91 & 1576.17\\
Int + slope & Spec & 50 & 4 & 0.96 & 1719.74 & 59.87 & 0.98 & 1754.84 & 61.09 & 1.10 & 1954.25 & 68.03 & 0.99 & 1763.77 & 61.40\\
Int + slope & Spec & 50 & 6 & 0.96 & 545.70 & 2.34 & 0.98 & 556.83 & 2.39 & 1.10 & 620.11 & 2.66 & 0.99 & 559.67 & 2.40\\
Int + slope & Spec & 50 & 8 & 0.96 & 143.82 & 0.97 & 0.98 & 146.76 & 0.99 & 1.10 & 163.44 & 1.10 & 0.99 & 147.51 & 1.00\\
\hline
Int + slope & Spec & 50 & 10 & 0.96 & 7.20 & 0.96 & 0.98 & 7.35 & 0.98 & 1.10 & 8.18 & 1.09 & 0.99 & 7.38 & 0.99\\
Int + slope & Spec & 50 & 12 & 0.96 & 2.20 & 0.97 & 0.98 & 2.25 & 0.99 & 1.10 & 2.50 & 1.10 & 0.99 & 2.26 & 0.99\\
Int + slope & Spec & 50 & 16 & 0.96 & 0.97 & 0.97 & 0.98 & 0.98 & 0.99 & 1.10 & 1.10 & 1.10 & 0.99 & 0.99 & 0.99\\
Int + slope & Spec & 50 & 20 & 0.96 & 0.97 & 0.96 & 0.98 & 0.99 & 0.98 & 1.10 & 1.10 & 1.10 & 0.99 & 0.99 & 0.99\\
Int + slope & Spec & 50 & Inf & 0.96 & 0.96 & 0.96 & 0.98 & 0.98 & 0.98 & 1.10 & 1.10 & 1.10 & 0.99 & 0.99 & 0.99\\
\hline
Int + slope & Spec & 100 & 1 & 0.98 & 5265.70 & 3778.03 & 0.99 & 5318.89 & 3816.19 & 1.04 & 5601.81 & 4019.18 & 0.99 & 5332.38 & 3825.87\\
Int + slope & Spec & 100 & 2 & 0.98 & 3355.28 & 432.65 & 0.99 & 3389.17 & 437.02 & 1.04 & 3569.45 & 460.27 & 0.99 & 3397.77 & 438.13\\
Int + slope & Spec & 100 & 4 & 0.98 & 748.54 & 3.47 & 0.99 & 756.10 & 3.51 & 1.04 & 796.32 & 3.70 & 0.99 & 758.01 & 3.52\\
Int + slope & Spec & 100 & 6 & 0.98 & 149.23 & 0.98 & 0.99 & 150.73 & 0.99 & 1.04 & 158.75 & 1.04 & 0.99 & 151.12 & 0.99\\
Int + slope & Spec & 100 & 8 & 0.98 & 7.41 & 0.98 & 0.99 & 7.49 & 0.99 & 1.04 & 7.88 & 1.04 & 0.99 & 7.50 & 0.99\\
\hline
Int + slope & Spec & 100 & 10 & 0.98 & 0.98 & 0.98 & 0.99 & 0.99 & 0.99 & 1.04 & 1.05 & 1.04 & 0.99 & 1.00 & 0.99\\
Int + slope & Spec & 100 & 12 & 0.98 & 0.98 & 0.98 & 0.99 & 0.99 & 0.99 & 1.04 & 1.04 & 1.04 & 0.99 & 0.99 & 0.99\\
Int + slope & Spec & 100 & 16 & 0.98 & 0.98 & 0.98 & 0.99 & 0.99 & 0.99 & 1.04 & 1.04 & 1.04 & 0.99 & 0.99 & 0.99\\
Int + slope & Spec & 100 & 20 & 0.98 & 0.98 & 0.98 & 0.99 & 0.99 & 0.99 & 1.04 & 1.04 & 1.04 & 0.99 & 0.99 & 0.99\\
Int + slope & Spec & 100 & Inf & 0.98 & 0.98 & 0.98 & 0.99 & 0.99 & 0.99 & 1.04 & 1.04 & 1.04 & 0.99 & 0.99 & 0.99\\
\hline
Int + slope & Spec & 200 & 1 & 1.00 & 1523.12 & 1492.43 & 1.01 & 1530.78 & 1499.93 & 1.03 & 1570.23 & 1538.59 & 1.01 & 1532.71 & 1501.82\\
Int + slope & Spec & 200 & 2 & 1.00 & 1658.23 & 91.29 & 1.01 & 1666.57 & 91.74 & 1.03 & 1709.52 & 94.11 & 1.01 & 1668.67 & 91.86\\
Int + slope & Spec & 200 & 4 & 1.00 & 214.55 & 1.00 & 1.01 & 215.63 & 1.01 & 1.03 & 221.19 & 1.04 & 1.01 & 215.90 & 1.01\\
Int + slope & Spec & 200 & 6 & 1.00 & 7.43 & 1.01 & 1.01 & 7.47 & 1.01 & 1.03 & 7.66 & 1.04 & 1.01 & 7.48 & 1.01\\
Int + slope & Spec & 200 & 8 & 1.00 & 0.99 & 1.00 & 1.01 & 1.00 & 1.00 & 1.03 & 1.02 & 1.03 & 1.01 & 1.00 & 1.00\\
\hline
Int + slope & Spec & 200 & 10 & 1.00 & 1.01 & 1.00 & 1.01 & 1.01 & 1.00 & 1.03 & 1.04 & 1.03 & 1.01 & 1.01 & 1.00\\
Int + slope & Spec & 200 & 12 & 1.00 & 1.00 & 1.00 & 1.01 & 1.01 & 1.01 & 1.03 & 1.03 & 1.03 & 1.01 & 1.01 & 1.01\\
Int + slope & Spec & 200 & 16 & 1.00 & 1.00 & 1.00 & 1.01 & 1.01 & 1.00 & 1.03 & 1.03 & 1.03 & 1.01 & 1.01 & 1.00\\
Int + slope & Spec & 200 & 20 & 1.00 & 1.00 & 1.00 & 1.01 & 1.01 & 1.00 & 1.03 & 1.03 & 1.03 & 1.01 & 1.01 & 1.01\\
Int + slope & Spec & 200 & Inf & 1.00 & 1.00 & 1.00 & 1.01 & 1.01 & 1.01 & 1.03 & 1.03 & 1.03 & 1.01 & 1.01 & 1.01\\
\hline
Int & Misspec & 20 & 1 & 0.87 & 2606.01 & 0.87 & 0.91 & 2743.17 & 0.91 & 1.24 & 3722.87 & 1.24 & 0.92 & 2778.13 & 0.92\\
Int & Misspec & 20 & 2 & 0.87 & 1804.22 & 0.87 & 0.91 & 1899.18 & 0.91 & 1.24 & 2577.46 & 1.24 & 0.92 & 1923.39 & 0.92\\
Int & Misspec & 20 & 4 & 0.87 & 2049.65 & 0.87 & 0.91 & 2157.53 & 0.91 & 1.24 & 2928.07 & 1.24 & 0.92 & 2185.03 & 0.92\\
Int & Misspec & 20 & 6 & 0.87 & 593.35 & 0.87 & 0.91 & 624.57 & 0.91 & 1.24 & 847.64 & 1.24 & 0.92 & 632.54 & 0.92\\
Int & Misspec & 20 & 8 & 0.87 & 212.87 & 0.87 & 0.91 & 224.08 & 0.91 & 1.24 & 304.11 & 1.24 & 0.92 & 226.93 & 0.92\\
\hline
Int & Misspec & 20 & 10 & 0.87 & 45.27 & 0.87 & 0.91 & 47.65 & 0.91 & 1.24 & 64.67 & 1.24 & 0.92 & 48.26 & 0.92\\
Int & Misspec & 20 & 12 & 0.87 & 0.87 & 0.87 & 0.91 & 0.92 & 0.91 & 1.24 & 1.25 & 1.24 & 0.92 & 0.93 & 0.92\\
Int & Misspec & 20 & 16 & 0.87 & 0.87 & 0.87 & 0.91 & 0.92 & 0.91 & 1.24 & 1.25 & 1.24 & 0.92 & 0.93 & 0.92\\
Int & Misspec & 20 & 20 & 0.87 & 0.87 & 0.87 & 0.91 & 0.91 & 0.91 & 1.24 & 1.24 & 1.24 & 0.92 & 0.93 & 0.92\\
Int & Misspec & 20 & Inf & 0.87 & 0.87 & 0.87 & 0.91 & 0.91 & 0.91 & 1.24 & 1.24 & 1.24 & 0.92 & 0.92 & 0.92\\
\hline
Int & Misspec & 50 & 1 & 0.95 & 4619.90 & 0.95 & 0.97 & 4714.18 & 0.97 & 1.08 & 5249.89 & 1.08 & 0.98 & 4738.10 & 0.98\\
Int & Misspec & 50 & 2 & 0.95 & 3783.15 & 0.95 & 0.97 & 3860.35 & 0.97 & 1.08 & 4299.03 & 1.08 & 0.98 & 3879.94 & 0.98\\
Int & Misspec & 50 & 4 & 0.95 & 515.13 & 0.95 & 0.97 & 525.65 & 0.97 & 1.08 & 585.38 & 1.08 & 0.98 & 528.31 & 0.98\\
Int & Misspec & 50 & 6 & 0.95 & 59.02 & 0.95 & 0.97 & 60.22 & 0.97 & 1.08 & 67.07 & 1.08 & 0.98 & 60.53 & 0.98\\
Int & Misspec & 50 & 8 & 0.95 & 0.96 & 0.95 & 0.97 & 0.98 & 0.97 & 1.08 & 1.09 & 1.08 & 0.98 & 0.98 & 0.98\\
\hline
Int & Misspec & 50 & 10 & 0.95 & 0.96 & 0.95 & 0.97 & 0.98 & 0.97 & 1.08 & 1.09 & 1.08 & 0.98 & 0.99 & 0.98\\
Int & Misspec & 50 & 12 & 0.95 & 0.96 & 0.95 & 0.97 & 0.97 & 0.97 & 1.08 & 1.09 & 1.08 & 0.98 & 0.98 & 0.98\\
Int & Misspec & 50 & 16 & 0.95 & 0.95 & 0.95 & 0.97 & 0.97 & 0.97 & 1.08 & 1.08 & 1.08 & 0.98 & 0.98 & 0.98\\
Int & Misspec & 50 & 20 & 0.95 & 0.95 & 0.95 & 0.97 & 0.97 & 0.97 & 1.08 & 1.08 & 1.08 & 0.98 & 0.98 & 0.98\\
Int & Misspec & 50 & Inf & 0.95 & 0.95 & 0.95 & 0.97 & 0.97 & 0.97 & 1.08 & 1.08 & 1.08 & 0.98 & 0.98 & 0.98\\
\hline
Int & Misspec & 100 & 1 & 0.98 & 6064.39 & 0.98 & 0.99 & 6125.65 & 0.99 & 1.04 & 6451.48 & 1.04 & 0.99 & 6141.18 & 0.99\\
Int & Misspec & 100 & 2 & 0.98 & 1944.58 & 0.98 & 0.99 & 1964.22 & 0.99 & 1.04 & 2068.70 & 1.04 & 0.99 & 1969.20 & 0.99\\
Int & Misspec & 100 & 4 & 0.98 & 110.96 & 0.98 & 0.99 & 112.09 & 0.99 & 1.04 & 118.05 & 1.04 & 0.99 & 112.37 & 0.99\\
Int & Misspec & 100 & 6 & 0.98 & 0.98 & 0.98 & 0.99 & 0.99 & 0.99 & 1.04 & 1.04 & 1.04 & 0.99 & 0.99 & 0.99\\
Int & Misspec & 100 & 8 & 0.98 & 0.98 & 0.98 & 0.99 & 0.99 & 0.99 & 1.04 & 1.04 & 1.04 & 0.99 & 0.99 & 0.99\\
\hline
Int & Misspec & 100 & 10 & 0.98 & 0.99 & 0.98 & 0.99 & 1.00 & 0.99 & 1.04 & 1.05 & 1.04 & 0.99 & 1.00 & 0.99\\
Int & Misspec & 100 & 12 & 0.98 & 0.98 & 0.98 & 0.99 & 0.99 & 0.99 & 1.04 & 1.04 & 1.04 & 0.99 & 0.99 & 0.99\\
Int & Misspec & 100 & 16 & 0.98 & 0.98 & 0.98 & 0.99 & 0.99 & 0.99 & 1.04 & 1.04 & 1.04 & 0.99 & 0.99 & 0.99\\
Int & Misspec & 100 & 20 & 0.98 & 0.98 & 0.98 & 0.99 & 0.99 & 0.99 & 1.04 & 1.04 & 1.04 & 0.99 & 0.99 & 0.99\\
Int & Misspec & 100 & Inf & 0.98 & 0.98 & 0.98 & 0.99 & 0.99 & 0.99 & 1.04 & 1.04 & 1.04 & 0.99 & 0.99 & 0.99\\
\hline
Int & Misspec & 200 & 1 & 0.98 & 3709.06 & 0.98 & 0.98 & 3727.70 & 0.98 & 1.01 & 3823.77 & 1.01 & 0.99 & 3732.41 & 0.98\\
Int & Misspec & 200 & 2 & 0.98 & 698.88 & 0.98 & 0.98 & 702.39 & 0.98 & 1.01 & 720.49 & 1.01 & 0.99 & 703.28 & 0.98\\
Int & Misspec & 200 & 4 & 0.98 & 0.97 & 0.98 & 0.98 & 0.98 & 0.98 & 1.01 & 1.00 & 1.01 & 0.99 & 0.98 & 0.99\\
Int & Misspec & 200 & 6 & 0.98 & 0.99 & 0.98 & 0.98 & 0.99 & 0.98 & 1.01 & 1.02 & 1.01 & 0.99 & 0.99 & 0.99\\
Int & Misspec & 200 & 8 & 0.98 & 0.99 & 0.98 & 0.98 & 0.99 & 0.98 & 1.01 & 1.02 & 1.01 & 0.99 & 0.99 & 0.99\\
\hline
Int & Misspec & 200 & 10 & 0.98 & 0.98 & 0.98 & 0.98 & 0.99 & 0.98 & 1.01 & 1.01 & 1.01 & 0.99 & 0.99 & 0.99\\
Int & Misspec & 200 & 12 & 0.98 & 0.98 & 0.98 & 0.98 & 0.99 & 0.98 & 1.01 & 1.01 & 1.01 & 0.99 & 0.99 & 0.99\\
Int & Misspec & 200 & 16 & 0.98 & 0.98 & 0.98 & 0.98 & 0.98 & 0.98 & 1.01 & 1.01 & 1.01 & 0.99 & 0.99 & 0.99\\
Int & Misspec & 200 & 20 & 0.98 & 0.98 & 0.98 & 0.98 & 0.98 & 0.98 & 1.01 & 1.01 & 1.01 & 0.99 & 0.98 & 0.99\\
Int & Misspec & 200 & Inf & 0.98 & 0.98 & 0.98 & 0.98 & 0.98 & 0.98 & 1.01 & 1.01 & 1.01 & 0.99 & 0.99 & 0.99\\
\hline
Int + slope & Misspec & 20 & 1 & 0.85 & 2846.71 & 0.85 & 0.90 & 2996.54 & 0.90 & 1.22 & 4066.73 & 1.22 & 0.91 & 3034.80 & 0.91\\
Int + slope & Misspec & 20 & 2 & 0.85 & 2001.34 & 0.85 & 0.90 & 2106.67 & 0.90 & 1.22 & 2859.05 & 1.22 & 0.91 & 2133.57 & 0.91\\
Int + slope & Misspec & 20 & 4 & 0.85 & 1990.16 & 0.85 & 0.90 & 2094.90 & 0.90 & 1.22 & 2843.08 & 1.22 & 0.91 & 2121.65 & 0.91\\
Int + slope & Misspec & 20 & 6 & 0.85 & 429.01 & 0.85 & 0.90 & 451.59 & 0.90 & 1.22 & 612.87 & 1.22 & 0.91 & 457.36 & 0.91\\
Int + slope & Misspec & 20 & 8 & 0.85 & 72.45 & 0.85 & 0.90 & 76.26 & 0.90 & 1.22 & 103.50 & 1.22 & 0.91 & 77.24 & 0.91\\
\hline
Int + slope & Misspec & 20 & 10 & 0.85 & 4.53 & 0.85 & 0.90 & 4.77 & 0.90 & 1.22 & 6.48 & 1.22 & 0.91 & 4.83 & 0.91\\
Int + slope & Misspec & 20 & 12 & 0.85 & 1.24 & 0.85 & 0.90 & 1.30 & 0.90 & 1.22 & 1.77 & 1.22 & 0.91 & 1.32 & 0.91\\
Int + slope & Misspec & 20 & 16 & 0.85 & 0.86 & 0.85 & 0.90 & 0.90 & 0.90 & 1.22 & 1.22 & 1.22 & 0.91 & 0.91 & 0.91\\
Int + slope & Misspec & 20 & 20 & 0.85 & 0.85 & 0.85 & 0.90 & 0.90 & 0.90 & 1.22 & 1.22 & 1.22 & 0.91 & 0.91 & 0.91\\
Int + slope & Misspec & 20 & Inf & 0.85 & 0.85 & 0.85 & 0.90 & 0.90 & 0.90 & 1.22 & 1.22 & 1.22 & 0.91 & 0.91 & 0.91\\
\hline
Int + slope & Misspec & 50 & 1 & 0.95 & 7515.61 & 0.95 & 0.97 & 7668.99 & 0.97 & 1.08 & 8540.47 & 1.08 & 0.97 & 7707.96 & 0.97\\
Int + slope & Misspec & 50 & 2 & 0.95 & 3159.57 & 0.95 & 0.97 & 3224.05 & 0.97 & 1.08 & 3590.42 & 1.08 & 0.97 & 3240.43 & 0.97\\
Int + slope & Misspec & 50 & 4 & 0.95 & 337.43 & 0.95 & 0.97 & 344.32 & 0.97 & 1.08 & 383.44 & 1.08 & 0.97 & 346.07 & 0.97\\
Int + slope & Misspec & 50 & 6 & 0.95 & 13.51 & 0.95 & 0.97 & 13.78 & 0.97 & 1.08 & 15.35 & 1.08 & 0.97 & 13.85 & 0.97\\
Int + slope & Misspec & 50 & 8 & 0.95 & 0.95 & 0.95 & 0.97 & 0.97 & 0.97 & 1.08 & 1.08 & 1.08 & 0.97 & 0.98 & 0.97\\
\hline
Int + slope & Misspec & 50 & 10 & 0.95 & 0.95 & 0.95 & 0.97 & 0.97 & 0.97 & 1.08 & 1.08 & 1.08 & 0.97 & 0.98 & 0.97\\
Int + slope & Misspec & 50 & 12 & 0.95 & 0.95 & 0.95 & 0.97 & 0.97 & 0.97 & 1.08 & 1.08 & 1.08 & 0.97 & 0.97 & 0.97\\
Int + slope & Misspec & 50 & 16 & 0.95 & 0.95 & 0.95 & 0.97 & 0.97 & 0.97 & 1.08 & 1.08 & 1.08 & 0.97 & 0.98 & 0.97\\
Int + slope & Misspec & 50 & 20 & 0.95 & 0.95 & 0.95 & 0.97 & 0.97 & 0.97 & 1.08 & 1.08 & 1.08 & 0.97 & 0.98 & 0.97\\
Int + slope & Misspec & 50 & Inf & 0.95 & 0.95 & 0.95 & 0.97 & 0.97 & 0.97 & 1.08 & 1.08 & 1.08 & 0.97 & 0.97 & 0.97\\
\hline
Int + slope & Misspec & 100 & 1 & 0.98 & 5695.08 & 0.98 & 0.99 & 5752.60 & 0.99 & 1.04 & 6058.59 & 1.04 & 0.99 & 5767.17 & 0.99\\
Int + slope & Misspec & 100 & 2 & 0.98 & 1598.79 & 0.98 & 0.99 & 1614.94 & 0.99 & 1.04 & 1700.84 & 1.04 & 0.99 & 1619.03 & 0.99\\
Int + slope & Misspec & 100 & 4 & 0.98 & 19.93 & 0.98 & 0.99 & 20.13 & 0.99 & 1.04 & 21.20 & 1.04 & 0.99 & 20.18 & 0.99\\
Int + slope & Misspec & 100 & 6 & 0.98 & 0.99 & 0.98 & 0.99 & 1.00 & 0.99 & 1.04 & 1.05 & 1.04 & 0.99 & 1.00 & 0.99\\
Int + slope & Misspec & 100 & 8 & 0.98 & 0.97 & 0.98 & 0.99 & 0.98 & 0.99 & 1.04 & 1.03 & 1.04 & 0.99 & 0.99 & 0.99\\
\hline
Int + slope & Misspec & 100 & 10 & 0.98 & 0.98 & 0.98 & 0.99 & 0.99 & 0.99 & 1.04 & 1.04 & 1.04 & 0.99 & 0.99 & 0.99\\
Int + slope & Misspec & 100 & 12 & 0.98 & 0.98 & 0.98 & 0.99 & 0.99 & 0.99 & 1.04 & 1.04 & 1.04 & 0.99 & 0.99 & 0.99\\
Int + slope & Misspec & 100 & 16 & 0.98 & 0.98 & 0.98 & 0.99 & 0.99 & 0.99 & 1.04 & 1.04 & 1.04 & 0.99 & 0.99 & 0.99\\
Int + slope & Misspec & 100 & 20 & 0.98 & 0.98 & 0.98 & 0.99 & 0.99 & 0.99 & 1.04 & 1.04 & 1.04 & 0.99 & 0.99 & 0.99\\
Int + slope & Misspec & 100 & Inf & 0.98 & 0.98 & 0.98 & 0.99 & 0.99 & 0.99 & 1.04 & 1.04 & 1.04 & 0.99 & 0.99 & 0.99\\
\hline
Int + slope & Misspec & 200 & 1 & 1.00 & 4112.56 & 1.00 & 1.00 & 4133.22 & 1.00 & 1.03 & 4239.75 & 1.03 & 1.00 & 4138.44 & 1.00\\
Int + slope & Misspec & 200 & 2 & 1.00 & 561.51 & 1.00 & 1.00 & 564.34 & 1.00 & 1.03 & 578.88 & 1.03 & 1.00 & 565.05 & 1.00\\
Int + slope & Misspec & 200 & 4 & 1.00 & 1.00 & 1.00 & 1.00 & 1.00 & 1.00 & 1.03 & 1.03 & 1.03 & 1.00 & 1.01 & 1.00\\
Int + slope & Misspec & 200 & 6 & 1.00 & 1.01 & 1.00 & 1.00 & 1.01 & 1.00 & 1.03 & 1.04 & 1.03 & 1.00 & 1.01 & 1.00\\
Int + slope & Misspec & 200 & 8 & 1.00 & 0.99 & 1.00 & 1.00 & 1.00 & 1.00 & 1.03 & 1.02 & 1.03 & 1.00 & 1.00 & 1.00\\
\hline
Int + slope & Misspec & 200 & 10 & 1.00 & 1.00 & 1.00 & 1.00 & 1.00 & 1.00 & 1.03 & 1.03 & 1.03 & 1.00 & 1.00 & 1.00\\
Int + slope & Misspec & 200 & 12 & 1.00 & 1.00 & 1.00 & 1.00 & 1.00 & 1.00 & 1.03 & 1.03 & 1.03 & 1.00 & 1.00 & 1.00\\
Int + slope & Misspec & 200 & 16 & 1.00 & 1.00 & 1.00 & 1.00 & 1.00 & 1.00 & 1.03 & 1.03 & 1.03 & 1.00 & 1.00 & 1.00\\
Int + slope & Misspec & 200 & 20 & 1.00 & 1.00 & 1.00 & 1.00 & 1.00 & 1.00 & 1.03 & 1.03 & 1.03 & 1.00 & 1.00 & 1.00\\
Int + slope & Misspec & 200 & Inf & 1.00 & 1.00 & 1.00 & 1.00 & 1.00 & 1.00 & 1.03 & 1.03 & 1.03 & 1.00 & 1.00 & 1.00\\*
\end{longtable}
\endgroup{}
\end{landscape}

\section{Additional real data analysis}

\begin{table}[H]
\centering\begingroup\fontsize{9}{11}\selectfont
\caption{Quantiles of $L_2$ privacy cost in 10,000 times perturbation in the real data analysis.}
\begin{tabular}{cccccccccc}
\toprule
$\varepsilon_0$ & $Q_{0.01}$ & $Q_{0.05}$ & $Q_{0.1}$ & $Q_{0.25}$ & $Q_{0.5}$ & $Q_{0.75}$ & $Q_{0.9}$ & $Q_{0.95}$ & $Q_{0.99}$\\
\midrule
1 & 0.006 & 0.011 & 0.014 & 0.022 & 0.035 & 0.059 & 43.539 & 104.902 & 164.659\\
2 & 0.003 & 0.005 & 0.007 & 0.011 & 0.016 & 0.024 & 0.033 & 0.040 & 0.052\\
4 & 0.002 & 0.003 & 0.003 & 0.005 & 0.008 & 0.012 & 0.016 & 0.020 & 0.025\\
6 & 0.001 & 0.002 & 0.002 & 0.004 & 0.005 & 0.008 & 0.011 & 0.013 & 0.017\\
8 & 0.001 & 0.001 & 0.002 & 0.003 & 0.004 & 0.006 & 0.008 & 0.010 & 0.013\\
10 & 0.001 & 0.001 & 0.001 & 0.002 & 0.003 & 0.005 & 0.007 & 0.008 & 0.010\\
12 & 0.001 & 0.001 & 0.001 & 0.002 & 0.003 & 0.004 & 0.005 & 0.006 & 0.008\\
16 & 0.000 & 0.001 & 0.001 & 0.001 & 0.002 & 0.003 & 0.004 & 0.005 & 0.006\\
20 & 0.000 & 0.001 & 0.001 & 0.001 & 0.002 & 0.002 & 0.003 & 0.004 & 0.005\\
\bottomrule
\end{tabular}
\endgroup{}
\end{table}

\begin{table}[H]
\centering\begingroup\fontsize{9}{11}\selectfont

\caption{Quantiles of SE calibration error in 10,000 times perturbation in the real data analysis.}
\begin{tabular}{cccccccccc}
\toprule
$\varepsilon_0$ & $Q_{0.01}$ & $Q_{0.05}$ & $Q_{0.1}$ & $Q_{0.25}$ & $Q_{0.5}$ & $Q_{0.75}$ & $Q_{0.9}$ & $Q_{0.95}$ & $Q_{0.99}$\\
\midrule
1 & 1.600 & 1.702 & 1.761 & 1.873 & 2.035 & 2.292 & 371388.664 & 10358224.546 & 17365127.861\\
2 & 1.111 & 1.159 & 1.189 & 1.240 & 1.307 & 1.394 & 1.491 & 1.565 & 1.724\\
4 & 0.968 & 0.994 & 1.011 & 1.043 & 1.082 & 1.130 & 1.177 & 1.208 & 1.271\\
6 & 0.952 & 0.973 & 0.986 & 1.009 & 1.037 & 1.069 & 1.100 & 1.119 & 1.158\\
8 & 0.953 & 0.971 & 0.981 & 0.999 & 1.021 & 1.045 & 1.067 & 1.082 & 1.109\\
10 & 0.958 & 0.972 & 0.981 & 0.996 & 1.014 & 1.033 & 1.050 & 1.062 & 1.083\\
12 & 0.962 & 0.974 & 0.982 & 0.994 & 1.009 & 1.025 & 1.040 & 1.049 & 1.066\\
16 & 0.969 & 0.979 & 0.985 & 0.994 & 1.005 & 1.017 & 1.028 & 1.035 & 1.048\\
20 & 0.974 & 0.982 & 0.987 & 0.994 & 1.003 & 1.013 & 1.021 & 1.027 & 1.037\\
\bottomrule
\end{tabular}
\endgroup{}
\end{table}

\end{document}